\documentclass[11pt,letterpaper]{article}

\usepackage{authblk}

\usepackage[margin=1in]{geometry}
\usepackage{mathtools,amssymb,amsthm,mathrsfs}
\usepackage{comment}
\usepackage{booktabs}
\usepackage{graphicx}
\usepackage{enumitem}
\usepackage[dvipsnames]{xcolor}

\usepackage[
    backend=biber,
    style=alphabetic,
    sorting=none,
    maxalphanames=4,
    minalphanames=3
]{biblatex}
\usepackage{hyperref}
\hypersetup{
    colorlinks=true,
    linkcolor=blue,
    filecolor=blue,
    citecolor=blue,
    urlcolor=blue,
}
\theoremstyle{plain}
\newtheorem{theorem}{Theorem}
\newtheorem{lemma}[theorem]{Lemma}
\newtheorem{proposition}[theorem]{Proposition}
\newtheorem{corollary}[theorem]{Corollary}
\newtheorem{claim}{Claim}
\newtheorem{fact}{Fact}

\theoremstyle{definition}
\newtheorem{definition}[theorem]{Definition}

\newcommand{\comments}[1]{}
\newcommand{\mc}{\mathcal}
\newcommand{\mbb}{\mathbb}
\newcommand{\ket}[1]{\lvert#1\rangle}
\newcommand{\bra}[1]{\langle#1\rvert}
\newcommand{\ketbra}[2]{\ket{#1}\bra{#2}}

\newcommand{\Tr}{\operatorname{Tr}}
\newcommand{\im}{\operatorname{Im}}
\newcommand{\Cmagicd}[1]{\mathscr{C}_{\mathrm{magic}}^{(#1)}}

\setlist{leftmargin=1.8em,itemsep=4pt,topsep=5pt}
\allowdisplaybreaks[1]

\NewDocumentCommand{\hzy}{m}{\textcolor{Blue}{(hzy:\space#1)}}
\NewDocumentCommand{\zm}{m}{\textcolor{Brown}{(lzm:\space#1)}}

\begin{document}
\hypersetup{pageanchor=false}
\title{NLTM Hamiltonians from gauged sheaf quantum locally testable codes}
\renewcommand{\Authsep}{ , }
\renewcommand{\Authand}{ and }
\renewcommand{\Authands}{ , and }

\setlength{\affilsep}{0.6em}

\title{NLTM Hamiltonians from gauged sheaf quantum locally testable codes}

\author{Fuchuan Wei\hspace{0.1em}%
  \thanks{These authors contributed equally to this work.}}
\author{Zhengyi Han\hspace{0.1em}\protect\footnotemark[1]}
\author{Zimu Li\hspace{0.1em}\protect\footnotemark[1]}
\author{Zi-Wen Liu\hspace{0.1em}%
  \thanks{\texttt{zwliu0@tsinghua.edu.cn}}}

\affil{Yau Mathematical Sciences Center, Tsinghua University,
  Beijing 100084, China}
\date{\today}
\begin{titlepage}
\maketitle
\begin{abstract}

Understanding the complexity of low-energy quantum states is a central problem in quantum complexity theory and many-body physics. Here we prove the no low-energy trivial magic (NLTM) conjecture, constructing a family of local Hamiltonians for which every state within a positive energy density window requires circuit depth $\Omega(\log n)$ to prepare from arbitrary stabilizer state that may themselves be highly entangled. Notably, our result is intrinsic in the sense that the same Hamiltonian family and energy density threshold apply to arbitrary qudit stabilizer inputs with any bounded local dimensions. Our construction builds on non-Abelian gauged sheaf codes obtained using the cup product structure on good quantum locally testable codes. We establish constant operator soundness for the associated Hamiltonians, enabling their protected logical structure to constrain all states at sufficiently low energy density. By excluding shallow stabilizer-based classical witnesses, Our result substantially strengthens NLTS and advances our understanding of low-energy complexity relevant to quantum
PCP.

\end{abstract}
\thispagestyle{empty}
\end{titlepage}
\hypersetup{pageanchor=true}
\pagenumbering{roman}
\tableofcontents
\clearpage
\pagenumbering{arabic}

\section{Introduction}

The complexity of low-energy quantum states is a central theme in quantum many-body physics and quantum complexity theory.
A culminating question with far-reaching implications for quantum
complexity theory and our understanding of quantum matter is captured by the quantum probabilistically checkable proofs (qPCP) conjecture, which posits that approximating the ground-state energy density of
a local Hamiltonian to constant accuracy remains
$\mathsf{QMA}$-hard~\cite{aharonov2013quantumpcpconjecture}, thus connecting quantum proof verification and hardness of approximation with the structure of low-energy many-body states.
A key obstacle is that low-energy states may admit succinct classical
descriptions that enable efficient estimation of their energies,
thereby providing classical witnesses for low energy~\cite{brandão2014productstateapproximationsquantumground,gharibian2024dequantizingquantumsingularvalue}.
Understanding how local Hamiltonians can rule out increasingly broad
classes of such witnesses is therefore central to progress toward qPCP.

A key milestone in the quest for qPCP is the no low-energy trivial states (NLTS) conjecture.
Formulated by Freedman and Hastings~\cite{freedman2013quantumsystemsnonkhyperfinitecomplexes},
NLTS asks for families of local Hamiltonians whose low-energy
states cannot be prepared from product states by constant-depth circuits, giving a physically motivated intermediate goal toward qPCP.
Building on breakthroughs in the construction of asymptotically good
quantum low-density parity-check (qLDPC) codes~\cite{PK2022Good,DHLV2022,leverrier2022quantumtannercodes},
Anshu, Breuckmann, and Nirkhe established the existence of NLTS local Hamiltonians~\cite{anshu2024nltshamiltoniansgoodquantum}.
Subsequent work extended this approach to suitable low-rate qLDPC
codes~\cite{golowich2023nltshamiltoniansstronglyexplicitsos}.

However, NLTS excludes only a specific type of classical witnesses, namely shallow circuits acting on product states. Crucial to our understanding of quantum computational advantage is
the observation that stabilizer states can be highly entangled yet
admit compact classical descriptions and efficiently computable local
expectation values, as follows from the Gottesman--Knill
theorem~\cite{gottesman1998heisenberg,aaronson2008improvedsimulationstabilizercircuits}. Indeed, the NLTS Hamiltonians admit stabilizer ground states, whose
compact classical descriptions provide efficiently checkable
witnesses for low energy.

This brings into play quantum magic (nonstabilizerness)~\cite{bravyi2004universalquantumcomputationideal,veitch2013resourcetheorystabilizercomputation,howard2014contextualitysuppliesmagicquantum,liu2022manybodyquantummagic},
which has attracted growing interest in quantum information and many-body physics (see e.g.\ Refs.\ \cite{liu2022manybodyquantummagic,leone2022stabilizerrenyientropy,haug2023quantifyingnonstabilizernessmatrixproduct,white2021conformalfieldtheoriesmagical,ellison2021symmetryprotectedsignproblem,gu2025magicinducedcomputational}).
Of particular relevance in the present Hamiltonian complexity context is the notion of \emph{long-range magic} (LRM), namely magic that
cannot be removed by shallow circuits, which has also been developed in connection with quantum error correction and phases of matter~\cite{korbany2025longrangenonstabilizernessphasesmatter,wei2025longrangenonstabilizernessquantumcodes}.
From this perspective, advancing toward qPCP requires strengthening
NLTS to exclude stabilizer-based witnesses in addition to the shallow entanglement type.
This naturally motivates the \emph{no low-energy trivial magic}
(NLTM) conjecture~\cite{wei2025longrangenonstabilizernessquantumcodes} (see also Refs.~\cite{natarajan2025lecturequantumpcpconjecture,parham2025quantumcircuitlowerbounds}) that asks for families of local Hamiltonians for which every state within a linear energy (or positive energy density) window above the ground state
exhibits LRM, thereby excluding classical witnesses based on shallow preparation from arbitrary stabilizer inputs, rather than only from product states.

Relevant progress along this line includes recent results on magic
circuit lower bounds for encoded states, non-Abelian ground spaces,
and explicit state families~\cite{wei2025longrangenonstabilizernessquantumcodes,parham2025quantumcircuitlowerbounds,li2026explicitstatestwosidedlongrange}.
However, extending such obstructions uniformly to a constant energy density window as required by NLTM is an essentially different challenge.
At positive energy density, there are constructions that exclude
stabilizer states from the low-energy sector~\cite{coble2023localhamiltonianslowenergystabilizer}
and require an extensive number of non-Clifford gates to prepare
any low-energy state~\cite{coble2024hamiltonianslowenergystatesrequire},
but these constructions still admit ground states obtainable from
stabilizer inputs by a single layer of local rotations.
In addition, extensive magic has been shown to persist under
geometrically local circuits for low-energy pure states of non-Abelian string-net models, but the admissible energy density depends on circuit depth~\cite{zhang2026extensivelongrangemagicnonabelian}.
All these results leave fundamental gaps to NLTM, which requires strong uniform control over magic circuit depth
throughout a positive energy density window, beyond what ground space LRM or code distance can guarantee.

In this work, we introduce and prove two natural notions of NLTM (as illustrated in Fig.~\ref{figure:nlts-nltm-preparation}):
binary NLTM rules out shallow preparation from qubit stabilizer states, whereas intrinsic NLTM is an alphabet-independent strengthening that achieves the same obstruction for any stabilizer states on qudits with local dimension \(d_v\le d_{\max}\) for any fixed \(d_{\max}\). 
While the former notion already represents a natural form of NLTM, the refined form of the latter intrinsic version has essential motivations in relation to complexity theory and qPCP.
A classical witness for a low-energy state of a qubit Hamiltonian may use a stabilizer state on higher-dimensional qudits.
When the circuit depth and the qudit dimensions are independent of system size, the energy of the Hamiltonian can still be evaluated efficiently from classical descriptions of the stabilizer state and the circuit.
Intrinsic NLTM  rules out such witnesses for qudits of any
bounded dimension, rather than only those based on qubit stabilizer states.

We prove in Appendix~\ref{appendix:ternary} that the binary NLTM theorem can be established by encoding each qutrit of a ternary quantum Tanner code into two qubits.
Since the code has a linear Pauli energy barrier, its low-energy
subspace splits into logical and excitation spaces with the logical space having dimension \(3^k\).
A shallow preparation from a qubit stabilizer state would then produce a qubit stabilizer code \(\mathcal K\) such that \(3^k\) divides
\(\dim\mathcal K\), contradicting the fact that
\(\dim\mathcal K\) is a power of two.

Our main result is the intrinsic NLTM theorem:
\begin{theorem}[Intrinsic NLTM, informal]\label{theorem:intrinsic-informal}
There is a family of frustration-free qubit local Hamiltonians, such that, for every given bound on the local dimensions, preparing any state below a positive energy density from a qudit stabilizer state requires depth $\Omega(\log n)$.

\end{theorem}

Note again that we allow arbitrary, potentially highly entangled stabilizer states on qudits with \(2\le d_v\le d_{\max}\), including  states
with unequal or composite local dimensions.
The state may contain any finite number of ancillary qudits, and its Clifford circuit depth is unrestricted.
Ancillary qudits may be discarded, and  classical mixtures are
also allowed.
Fig.~\ref{figure:nlts-nltm-preparation} summarizes the distinctions among NLTS, binary NLTM, and intrinsic NLTM.

\begin{figure}
\centering
\includegraphics[width=\textwidth]{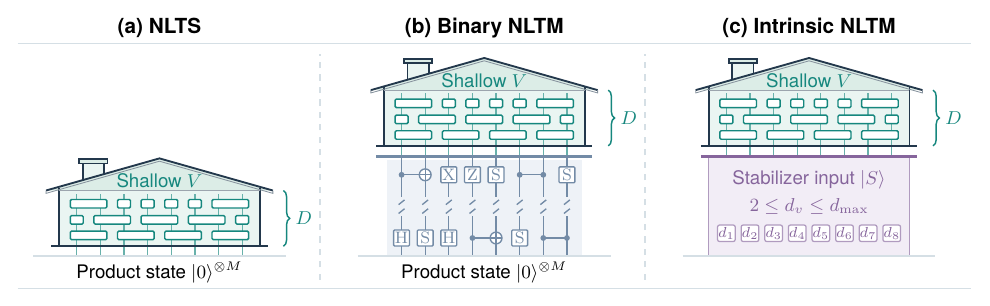}
\caption{\textbf{Preparation models excluded at low energy.}
The bungalow in (a) and the stilt houses in (b,c) visualize the extra free stabilizer structure beneath $V$ in NLTM.
(a) NLTS excludes shallow circuits $V$ acting on product inputs.
(b) Binary NLTM permits arbitrary qubit stabilizer inputs, including those with long-range entanglement.
(c) Intrinsic NLTM permits stabilizer inputs on qudits of dimensions $2\le d_v\le d_{\max}$, for every fixed $d_{\max}$.
The dimensions may be unequal or composite.
Only the depth of $V$ is counted.
Gates act on at most two qudits with arbitrary connectivity, and each qudit's dimension remains fixed throughout the circuit.
All panels show purifications before ancillary outputs are discarded.}
\label{figure:nlts-nltm-preparation}
\end{figure}

We now sketch the key insights and steps of our proof.
Suppose that an low energy state \(\rho\)  can be prepared by a shallow circuit \(V\) from a stabilizer state \(\ket{S}\).
Write \(H=\sum_i h_i\).
We can construct a stabilizer code \(\mathcal K\) whose stabilizer group contains the elements of \(\operatorname{Stab}(S)\) supported within the regions \(\operatorname{supp}(V^\dagger h_iV)\).
Every state in $V\mc{K}$ has the same energy expectation with respect to $H$, no greater than that of $V\ket{S}$, and every prime divisor of $\dim\mc{K}$ is at most $d_{\max}$.
We choose $H$ so that a prime greater than $d_{\max}$ divides $\dim\mc{C}$, where $\mc{C}=\ker H$.
It remains to prove that $\dim\mc{C}$ divides $\dim\mc{K}$, which leads to a contradiction.


Let \(V_t\) be the span of states obtained by applying errors of weight at most \(t\) to \(\mathcal C\).
Linear distance gives a decomposition
\(V_t\simeq\mathcal C\otimes\mathcal G_t\), under which every operator with support below the distance bound acts trivially on \(\mathcal C\). When \(V\) is shallow, we construct an operator \(F\) that approximates the projector onto \(V\mathcal K\) and acts on \(V_t\) as \(I_{\mathcal C}\otimes B\).
Thus every eigenvalue of \(F\) on \(V_t\) has multiplicity divisible by \(\dim\mathcal C\).
If \(V\mathcal K\) were contained in \(V_t\), exactly \(\dim\mathcal K\) eigenvalues would be close to one, and hence
\(\dim\mathcal C\) would divide \(\dim\mathcal K\).

At positive energy density, however, \(V\mathcal K\) need not be contained in \(V_t\).
We therefore need to bound its projection onto \(V_t^\perp\), which requires a property analogous to the soundness of a good quantum locally testable code (qLTC)~\cite{gay2026asymptoticallygoodquantumlocally,BLN2026,CHLT2026}.
A good CSS code alone does not suffice, since its code space has dimension a power of two. We instead construct a family of non-Abelian codes by gauging good qLTCs using the cup product structure recently established in Refs.~\cite{li2026transversalnoncliffordgatesgood,LWHL2609nonAbelian}, and prove that the corresponding Hamiltonians satisfy \emph{operator soundness}, which
gives the required bound on the projector onto \(V_t^\perp\).

This paper is organized as follows.
Section~\ref{section:statement} defines magic circuit complexity and
states our main results.
Section~\ref{section:arithmetic-obstruction} derives the magic circuit lower
bound from linear distance, operator soundness, and the dimension of
the ground space.
Section~\ref{section:intrinsic-construction} constructs a family of
gauged codes and proves that they have required properties.
Section~\ref{section:discussion} discusses the implications of our results and several open questions.
The binary NLTM construction is given in
Appendix~\ref{appendix:ternary}.

\section{Setting and main results}\label{section:statement}

In this section, we define magic circuit complexity and state our main results on intrinsic NLTM.


\subsection{Magic circuit complexity}\label{section:mgaic_complexity}

Let $d_{\max}\ge2$ be an integer upper bound on the local qudit dimensions.
We use the qudit stabilizer formalism on a finite collection of $M$ qudits with Hilbert space
\begin{equation}
\mc{H}=\bigotimes_{v=1}^M\mbb{C}^{d_v},\quad 2\le d_v\le d_{\max},
\end{equation}
where the integer local dimensions $d_v$ may be unequal or composite.
On qudit $v$, the shift and clock operators act as
\begin{equation}
X_v\ket{a}=\ket{a+1\bmod d_v},\quad Z_v\ket{a}=e^{2\pi ia/d_v}\ket{a}
\end{equation}
for $a=0,\cdots,d_v-1$.
A generalized Pauli operator is a scalar phase times $\bigotimes_v X_v^{a_v}Z_v^{b_v}$, where $a_v,b_v\in\mbb{Z}_{d_v}$.\footnote{For composite prime power dimensions, this cyclic Weyl convention differs from the finite field Pauli convention based on the field trace commonly used for qudit CSS codes.}
A \emph{pure qudit stabilizer state} is the unique common $+1$ eigenstate, up to an overall phase, of some finite abelian group of Pauli operators containing no nonidentity scalar.
Such a state can be entangled across qudits, including qudits of unequal dimensions.

A depth-$D$ circuit consists of $D\in\mbb{Z}_{\ge0}$ layers of arbitrary unitary gates on mutually disjoint sets of at most two qudits, with arbitrary connectivity.
Gates acting on any fixed number of qudits can be decomposed into gates acting on at most two qudits, with a depth increase by a factor depending only on that number and the fixed local dimension bound.

\begin{definition}[Magic circuit complexity]
For a state $\rho$ on $N$ qubits, consider  preparations of the form
\begin{equation}\label{equation:intrinsic-preparation}
\rho=\sum_{j=1}^r p_j\Tr_{A_j}\bigl(V_j\ketbra{S_j}{S_j}V_j^\dagger\bigr).
\end{equation}
Here $r\ge1$ is an integer, $p_j>0$, and $\sum_jp_j=1$.
In each branch, $\ket{S_j}$ is any  qudit stabilizer state on a finite qudit system with local dimensions not exceeding $d_{\max}$.
The circuit $V_j$ consists of unitary gates on at most two qudits.
The dimension of each qudit remains fixed throughout the circuit.
Discarding all qudits in $A_j$ leaves $N$ designated qubits.
The \emph{magic circuit complexity} of $\rho$ at local dimension bound $d_{\max}$ is
\begin{equation}
\Cmagicd{d_{\max}}(\rho)\coloneqq\min\max_{1\le j\le r}\operatorname{depth}(V_j),
\end{equation}
where the minimum ranges over all preparations in Eq.~\eqref{equation:intrinsic-preparation} with dimension bound $d_{\max}$.
\end{definition}

The input size, local dimensions, stabilizer state, circuit, and retained qubits may vary between branches.
The stabilizer input may include an arbitrary finite number of ancillary qudits, with arbitrary stabilizer entanglement between the retained qubits and discarded qudits.

$\Cmagicd{d_{\max}}(\rho)=0$ if and only if $\rho\in\mathrm{STAB}_N$, where $\mathrm{STAB}_N$ denotes the convex hull of pure stabilizer states on the $N$ output qubits.
The case $d_{\max}=2$ corresponds to the preparation model defining binary NLTM.
Increasing $d_{\max}$ can only decrease the complexity.

For a Hamiltonian with polynomially many terms of bounded support and norm, fixed depth and a fixed $d_{\max}$ allow efficient classical energy evaluation given polynomial descriptions of the stabilizer input and circuit.
Each local term has a backward light cone containing only a bounded number of qudits of bounded dimension, so its generalized Pauli expansion can be evaluated efficiently on the stabilizer input.
Even with unrestricted finite ancillary inputs, discarded outputs, and classical randomization, every such preparation admits a polynomial classical energy witness of no greater energy, up to an arbitrarily small inverse polynomial tolerance (Appendix~\ref{appendix:classical-witnesses}).

\subsection{General criteria for intrinsic NLTM}\label{sec:intrinsic_criterion}

Let $H=\sum_{i=1}^m h_i$ be a frustration-free Hamiltonian on $N$ qubits with $0\leq h_i\leq I$ 
and nonzero ground space $\mc{C}=\ker H$.
The code distance $d(\mc{C})$ associated with $\mc{C}$ is the smallest support size of an operator $O$ for which $P_{\mc{C}}OP_{\mc{C}}$ is not a scalar multiple of $P_{\mc{C}}$, where $P_{\mc{C}}$ projects onto $\mc{C}$.
Define the linear spaces spanned by code states after local errors as
\begin{equation}
V_j\coloneqq\operatorname{span}\{E\ket{\psi}:\ket{\psi}\in\mc{C}, \operatorname{wt}(E)\le j\}.
\end{equation}
Here $0\le j\le N$, and $\Pi_{V_j}$ denotes the orthogonal projector onto $V_j$.
Restricting $E$ to Pauli errors gives the same space.
The distance operator is
\begin{equation}
\Delta_{\mc{C}}\coloneqq\sum_{j=0}^{N-1}(I-\Pi_{V_j}).
\end{equation}
We use the following operator soundness inequality, which expresses quantum local testability for local Hamiltonians~\cite{leverrier2022localtestabilityquantumcoding}:
\begin{equation}\label{equation:operator-soundness}
\frac{H}{m}\geq\sigma\frac{\Delta_{\mc{C}}}{N},\quad \sigma>0.
\end{equation}
Since the spaces $V_j$ are nested, this inequality bounds the weight outside $V_j$ for every low-energy state.
It applies to arbitrary coherent superpositions, which is needed when the spectral argument compares entire subspaces.

Combining these properties with a growing prime divisor of the ground space degeneracy gives the following criterion, proved in Section~\ref{section:intrinsic-magic-lower-bound}.

\begin{theorem}[Intrinsic magic complexity]\label{theorem:intrinsic}
Consider a family of Hamiltonians $H_n=\sum_{i=1}^{m_n}h_{n,i}$ on $n$ qubits, defined for an unbounded sequence of sizes $n$, with $0\le h_{n,i}\le I$ and uniformly bounded term supports.
Let $\mc{C}_n\coloneqq\ker H_n$ be their nonzero ground spaces.
Suppose there are constants $\Delta_*,\sigma_*>0$, independent of $n$, such that:
\begin{enumerate}[label=(\roman*),leftmargin=*,itemsep=2pt,topsep=3pt,parsep=0pt]
\item \textbf{Linear distance.} The code $\mc{C}_n$ has distance at least $\Delta_*n$.
\item \textbf{Local testability.} The Hamiltonian satisfies $\frac{H_n}{m_n}\geq\sigma_*\frac{\Delta_{\mc{C}_n}}{n}$.
\item \textbf{Growing prime factor.} $\dim\mc{C}_n$ has a prime divisor $p_n$ with $p_n\to\infty$.
\end{enumerate}
Let $0\le e<\sigma_*\Delta_*/2$ be an energy density threshold.
There is a constant $c_e>0$, independent of $n$ and $d_{\max}$, such that the following holds for all sufficiently large $n$:
For every integer $2\le d_{\max}<p_n$, every $n$-qubit state $\rho$ with $\Tr(H_n\rho)/m_n\le e$ satisfies
\begin{equation}\label{equation:intrinsic-complexity-bound}
\Cmagicd{d_{\max}}(\rho)>\frac15\log n-\frac15\log \log d_{\max}-c_e.
\end{equation}
Since $p_n\to\infty$, for every $d_{\max}$ independent of $n$
we have $\Cmagicd{d_{\max}}(\rho)\ge\Omega(\log n)$.
\end{theorem}

Linear distance (i) protects logical information within a correctable error space, while constant operator soundness (ii) bounds the weight of low-energy states outside this space.
Together, these properties keep states of sufficiently low energy density largely within a space where the logical information remains protected, sustaining the obstruction to shallow preparation.


\subsection{An intrinsic NLTM family}\label{sec:intrinsic_NLTM}



We present in Section~\ref{section:intrinsic-construction} a family of qubit gauged sheaf codes that combines positive rate, linear distance, and local testability with a growing prime divisor of the code dimension.

\begin{theorem}[Locally testable codes for intrinsic NLTM]\label{theorem:intrinsic-code}
There is a family of frustration-free qubit Hamiltonians $H_\nu$ on $N_\nu\to\infty$ qubits, with $m_\nu=\Theta(N_\nu)$ terms and bounded locality and degree, satisfying the hypotheses of Theorem~\ref{theorem:intrinsic}, including (i)--(iii).
Their ground spaces also have positive constant coding rate.
\end{theorem}

Our construction produces nonadditive qubit codes. Its constant coding rate is actually unnecessary in establishing the intrinsic NLTM theorem.
Combining Theorems~\ref{theorem:intrinsic} and~\ref{theorem:intrinsic-code} yields: 

\begin{corollary}[Intrinsic NLTM]\label{corollary:intrinsic-family}
There is a family of frustration-free qubit local Hamiltonians ${H_N}$, with $m_N=\Theta(N)$ local terms, 
and constants $\varepsilon,\kappa,c>0$, independent of $N$,
such that for every given $d_{\max}\ge2$
and all sufficiently large $N$, every $N$-qubit state $\rho$ with low energy density $\Tr(H_N\rho)/N\le\varepsilon$ obeys

\begin{equation}
\Cmagicd{d_{\max}}(\rho)>\kappa \log N-\frac15\log\log d_{\max}-c=\Omega(\log N).
\end{equation}
\end{corollary}


We remark that binary and intrinsic NLTM (Fig.~\ref{figure:nlts-nltm-preparation}) differ even at the level of ground-state preparation.
The ternary code constructed in Appendix~\ref{appendix:ternary} satisfies binary NLTM, but some of its ground states can be prepared in depth two from qutrit stabilizer inputs.
Their magic circuit complexity is thus at most two for $d_{\max}=3$, so this family does not satisfy intrinsic NLTM.

\section{Magic complexity lower bounds for low-energy states}
\label{section:intrinsic-magic-lower-bound}
\label{section:arithmetic-obstruction}

We prove Theorem~\ref{theorem:intrinsic} by combining logical protection from code distance, leakage control from operator soundness, and an arithmetic obstruction from the ground space dimension.
Suppose a shallow circuit maps a qudit stabilizer state to a low-energy state.
Stabilizer completion (Lemma~\ref{lemma:stabcompletion}) gives a nonzero qudit stabilizer code whose entire image under the circuit has the same mean energy, no greater than that of the original stabilizer state.
Soundness limits how much of a low-energy state lies outside the space generated by correctable errors acting on the ground space. Within this space, code distance protects all logical information encoded in the ground space.
For sufficiently shallow circuits, a polynomial approximation then forces the ground space dimension to divide the dimension of the completed stabilizer code.
This is impossible when the ground space dimension has a prime divisor larger than every local qudit dimension.

\subsection{Stabilizer completion}

We first extend the energy bound of a single stabilizer state to an entire stabilizer subspace.

\begin{lemma}[Stabilizer completion]\label{lemma:stabcompletion}
Let $H'=\sum_iO_i$ act on a finite qudit system with local dimensions $d_v$, where each $O_i$ is Hermitian and supported within $A_i$.
For every  qudit stabilizer state $\ket{S_0}$, there is a nonzero qudit stabilizer projector $P$, defined by $s$ commuting Pauli constraints each supported within some $A_i$, such that
\begin{equation}\label{equation:inf-uniform-energy}
PH'P=E_*P,\quad E_*\le\bra{S_0}H'\ket{S_0}.
\end{equation}
The rank of $P$ divides $\prod_v d_v$, and the generators can be chosen with $s\le\sum_{v\in\bigcup_iA_i}\log d_v$.
\end{lemma}

The construction retains the input stabilizers supported within individual regions $A_i$, extends them to a maximal commuting family of local constraints, and selects a joint eigenspace whose mean energy does not increase.
Maximality implies $PO_iP=c_iP$ for every $i$, where $c_i\in\mbb{R}$, and summing over $i$ gives $PH'P=E_*P$ with $E_*=\sum_i c_i$.
Thus every normalized state in $\im P$ has mean energy $E_*$, although $H'$ need not preserve this subspace.
Appendix~\ref{app:stabcomp} gives the proof.


\subsection{Logical protection and dimension constraints}


We next show that code distance gives the space generated by local errors a subsystem code structure, with the original code as a protected logical subsystem.
Let $\mc{C}$ be a nonzero code on $n$ qubits.
For an integer $0\le t\le n$, recall that $V_t$ is the span of Pauli errors of weight at most $t$ acting on $\mc{C}$, and let $\Pi_{V_t}$ be its orthogonal projector.

\begin{lemma}[Protected logical subsystem]\label{lemma:correctable-space-protection}
If $2t<d(\mc{C})$, there are a Hilbert space $\mc{G}_t$ and an onto isometry $W_t:\mc{C}\otimes\mc{G}_t\longrightarrow V_t$ such that every operator $O$ with $2t+|\operatorname{supp}O|<d(\mc{C})$ satisfies
\begin{equation}\label{equation:correctable-space-compression}
W_t^\dagger O W_t=I_{\mc{C}}\otimes B_O
\end{equation}
for some operator $B_O$ on $\mc{G}_t$.
\end{lemma}

Appendix~\ref{appendix:error-space-protection} proves this lemma.


We next use this protection to derive a dimension constraint.
Let $\widehat{H}$ be a sum of commuting orthogonal projectors acting on the $n$ qubits and a finite auxiliary system $A$.
The projector onto $\ker\widehat{H}$ need not be local, so we approximate it by a Chebyshev polynomial in $\widehat{H}$ of low degree~\cite{arad2013arealawsubexponentialalgorithm,anshu2024nltshamiltoniansgoodquantum}.
If every term in the expansion acts on fewer than $d(\mc{C})-2t$ of the $n$ qubits, its compression has the form $I_{\mc{C}}\otimes B$.
Under the overlap condition below, the compressed polynomial has exactly $\dim\ker\widehat{H}$ eigenvalues above the approximation error, counted with multiplicity.
The tensor form makes this count divisible by $\dim\mc{C}$.

\begin{lemma}[Rank divisibility]\label{lemma:intrinsic-rank-divisibility}
Let $\mc{C}\ne\{0\}$ be a code on $n$ qubits and $A$ a finite auxiliary system.
For an integer $0\le t\le n$, set $Q\coloneqq\Pi_{V_t}\otimes I_A$.
Let $\widehat{H}=\sum_{j=1}^s g_j$ have nonzero kernel, where the $g_j$ are commuting orthogonal projectors on these qubits and $A$, each acting on at most $w$ of the $n$ qubits.
For every $0<a<1/2$, there is a constant $c_a>0$ with the following property:
If
\begin{equation}\label{equation:intrinsic-rank-conditions}
\|\Pi_{\ker\widehat{H}}(I-Q)\Pi_{\ker\widehat{H}}\|_\infty<1-2a,\quad 2t+c_aw\sqrt{\max\{1,s\}}<d(\mc{C}),
\end{equation}
then
\begin{equation}\label{equation:intrinsic-rank-divisibility}
\dim\mc{C}\mid\dim\ker\widehat{H}.
\end{equation}
\end{lemma}

Appendix~\ref{appendix:refinement} gives the polynomial construction and the complete proof.

\subsection{Proof of the intrinsic NLTM criterion}

\begin{proof}
Let $0\le e<\sigma_*\Delta_*/2$.
We show that every sufficiently shallow preparation has energy density greater than $e$.
We first consider a single circuit branch and then extend the bound to finite mixtures.

Let $d_{\max}$ be an integer with $2\le d_{\max}<p_n$.
Suppose, for a contradiction, that
\begin{equation}\label{equation:intrinsic-assumed-energy}
\omega=\Tr_{A_{\mathrm{out}}}\bigl(V\ketbra{S_0}{S_0}V^\dagger\bigr),\quad E\coloneqq\Tr(H_n\omega)\le em_n,
\end{equation}
where $\ket{S_0}$ is a qudit stabilizer state with local dimensions at most $d_{\max}$, and $V$ has depth at most $D\in\mbb{Z}_{\ge0}$.
Set $H_{\mathrm{tot}}\coloneqq H_n\otimes I_{A_{\mathrm{out}}}$, and let $\ell\ge1$ be a uniform bound on the supports of the Hamiltonian terms.
For each term $h_{n,i}$, let $A_i$ be the backward light cone of its support under $V$.
Applying Lemma~\ref{lemma:stabcompletion} to $H'\coloneqq V^\dagger H_{\mathrm{tot}}V$ with these regions gives a nonzero stabilizer projector $P$ with $PH'P=E_*P$ for some $E_*\le E$.
Let $\Pi_1,\cdots,\Pi_s$ project onto the eigenspaces specified by the commuting Pauli constraints defining $P$, so that $P=\prod_{j=1}^s\Pi_j$.
Define
\begin{equation}
\widehat{H}\coloneqq\sum_{j=1}^s V(I-\Pi_j)V^\dagger.
\end{equation}
The summands of $\widehat{H}$ are commuting orthogonal projectors, $\Pi_{\ker\widehat{H}}=VPV^\dagger\ne0$, and
\begin{equation}
\Pi_{\ker\widehat{H}}H_{\mathrm{tot}}\Pi_{\ker\widehat{H}}=E_*\Pi_{\ker\widehat{H}}\leq E\Pi_{\ker\widehat{H}}.
\end{equation}

Each $A_i$ contains at most $\ell2^D$ input qudits.
Conjugation by $V$ therefore shows that each term of $\widehat{H}$ acts on at most $\ell2^{2D}$ qudits.
The union $\bigcup_iA_i$ lies within the backward light cones of the $n$ retained outputs and contains at most $n2^D$ input qudits.
The generator bound in Lemma~\ref{lemma:stabcompletion} yields $s\le n2^D\log d_{\max}$.
Moreover, $\dim\ker\widehat{H}=\operatorname{rank}P$ divides the total input Hilbert space dimension, so every prime divisor of $\dim\ker\widehat{H}$ is at most $d_{\max}$.

We next use soundness to verify the overlap condition in Lemma~\ref{lemma:intrinsic-rank-divisibility}.
Since $e<\sigma_*\Delta_*/2$, choose $0<a<1/2$ and $0<\alpha<\Delta_*/2$, depending only on $e$ and the fixed code parameters, such that $e<(1-2a)\sigma_*\alpha$.
Set $t\coloneqq\lfloor\alpha n\rfloor$.
For $0\le j\le t$, the inclusion $V_j\subset V_t$ gives $I-\Pi_{V_j}\geq I-\Pi_{V_t}$.
The first $t+1$ terms in the distance operator therefore imply
\begin{equation}\label{equation:distance-operator-tail}
\Delta_{\mc{C}_n}=\sum_{j=0}^{n-1}(I-\Pi_{V_j})\geq(t+1)(I-\Pi_{V_t}).
\end{equation}
Set $Q\coloneqq\Pi_{V_t}\otimes I_{A_{\mathrm{out}}}$.
By Eq.~\eqref{equation:operator-soundness} and Eq.~\eqref{equation:distance-operator-tail},
\begin{equation}
H_{\mathrm{tot}}\geq\frac{\sigma_*m_n}{n}\bigl(\Delta_{\mc{C}_n}\otimes I_{A_{\mathrm{out}}}\bigr)\geq\frac{\sigma_*m_n(t+1)}{n}(I-Q).
\end{equation}
Using the scalar compression of $H_{\mathrm{tot}}$ and the energy bound $E_*\le E\le em_n$, we obtain
\begin{align}
\begin{aligned}
    \Pi_{\ker\widehat{H}}(I-Q)\Pi_{\ker\widehat{H}}&\leq\frac{n}{\sigma_*m_n(t+1)}\Pi_{\ker\widehat{H}}H_{\mathrm{tot}}\Pi_{\ker\widehat{H}}\\
&=\frac{nE_*}{\sigma_*m_n(t+1)}\Pi_{\ker\widehat{H}}\\
&\leq\frac{ne}{\sigma_*(t+1)}\Pi_{\ker\widehat{H}}\\
&\leq\frac{e}{\sigma_*\alpha}\Pi_{\ker\widehat{H}}.\label{equation:intrinsic-output-overlap}
\end{aligned}
\end{align}
The last inequality uses $t+1>\alpha n$ and $e\ge0$.
Consequently,
\begin{equation}\label{equation:intrinsic-overlap-condition}
\|\Pi_{\ker\widehat{H}}(I-Q)\Pi_{\ker\widehat{H}}\|_\infty<1-2a.
\end{equation}

It remains to verify the support condition.
Set $\ell_D\coloneqq\ell2^{2D}$, which bounds the number of retained qubits on which each term of $\widehat{H}$ acts.
The generator bound gives
\begin{equation}\label{equation:intrinsic-support-bound}
c_a\ell_D\sqrt{\max\{1,s\}}\le c_a\ell\sqrt{n\log d_{\max}}\,2^{5D/2}.
\end{equation}
Choose $c_e>0$ large enough that $c_a\ell2^{-5c_e/2}<\Delta_*-2\alpha$.
This choice is independent of $n$ and $d_{\max}$.
If
\begin{equation}\label{equation:intrinsic-depth-budget}
D\le\frac15\log n-\frac15\log \log d_{\max}-c_e,
\end{equation}
then Eq.~\eqref{equation:intrinsic-support-bound} is strictly smaller than $(\Delta_*-2\alpha)n$.
Together with $2t\le2\alpha n$ and $d(\mc{C}_n)\ge\Delta_*n$, this yields
\begin{equation}
2t+c_a\ell_D\sqrt{\max\{1,s\}}<d(\mc{C}_n).
\end{equation}
Lemma~\ref{lemma:intrinsic-rank-divisibility} now implies $\dim\mc{C}_n\mid\dim\ker\widehat{H}$.
But $p_n>d_{\max}$ divides $\dim\mc{C}_n$, whereas every prime divisor of $\dim\ker\widehat{H}$ is at most $d_{\max}$, a contradiction.

Every single circuit branch of depth at most $D$ satisfying Eq.~\eqref{equation:intrinsic-depth-budget} therefore has energy density greater than $e$.
By linearity of the energy, the same strict bound holds for every finite mixture of such branches.
This proves \eqref{equation:intrinsic-complexity-bound}, with $c_e$ independent of $n$ and $d_{\max}$.
Since $p_n\to\infty$, the bound holds for every given $d_{\max}$ for all sufficiently large $n$.
\end{proof}


\section{Intrinsic NLTM from gauged sheaf codes}\label{section:intrinsic-construction}

We proof Theorem~\ref{theorem:intrinsic-code} in this section by using the gauged sheaf codes constructed in Appendix~\ref{appendix:gauged_code}. Its ground space dimension has a prime divisor that tends to infinity, while the local dimensions stays fixed and the relative distance and soundness have uniform positive lower bounds. The construction therefore satisfies the intrinsic NLTM criterion of Theorem~\ref{theorem:intrinsic}.

\subsection{Gauged sheaf code construction}

We start with three good quantum locally testable code blocks $A,B,C$ from the construction of Ref.~\cite{gay2026asymptoticallygoodquantumlocally}, with more detailed construction specified in Appendix~\ref{appendix:gauged_qLTC}. The number field $\mbb F_q$, where $q=2^{s_q}$, is fixed throughout the family. All three blocks have linear distance and constant soundness, and $A$ has positive rate.
The gauging construction of \cite{Zhu2026,christos2026nonabelianquantumlowdensityparity,LWHL2609nonAbelian,zhu2026nonAbelian} uses a trilinear invariant form $I_\xi$ on the cochains of these blocks. It is defined by cellular cup product as in Eq.~\eqref{eq:invariant_form}, and on cocycles, its value depends only on the cohomology classes (see also e.g., \cite{Li2025Poincare,LSWLL2026Theory,LLL2026nontrivial} for more details). Gauging leaves the $Z$ stabilizers unchanged and multiplies the $X$-stabilizer generators $\mathcal{A}^{(a)}_{\alpha v}$ with $a = A,B,C$ by phases determined by $I_\xi$: 
\begin{align}\label{eq:dressed_X_main}
	\begin{aligned}
		\mathcal{A}^{(A)}_{\alpha_A v_A} \ket{x_A,x_B,x_C}
		& = (-1)^{\Tr_{\mbb F_q/\mbb F_2} I_\xi(\alpha_A v_A, x_B,x_C)}
		\ket{x_A + \alpha_A \delta_A v_A ,x_B,x_C}, \\
		\mathcal{A}^{(B)}_{\alpha_B v_B} \ket{x_A,x_B,x_C}
		& = (-1)^{\Tr_{\mbb F_q/\mbb F_2} I_\xi(x_A, \alpha_B v_B,x_C)} \ket{x_A,x_B + \alpha_B \delta_B v_B, x_C}, \\
		\mathcal{A}^{(C)}_{\alpha_C v_C} \ket{x_A,x_B,x_C}
		& = (-1)^{\Tr_{\mbb F_q/\mbb F_2} I_\xi(x_A,x_B,\alpha_C v_C)} \ket{x_A,x_B,x_C + \alpha_C \delta_C v_C}.
	\end{aligned}
\end{align}
where $\alpha_a\in\mbb F_q$ for $a=A,B,C$, and $v_a$ and $x_a$ are cochains. The cochain complex language of sheaf codes is articulated in Appendix~\ref{appendix:gauged_intro}. All these stabilizer generators have constant weights and act on constant numbers of physical qudits. 

The gauged sheaf code in \cite{LWHL2609nonAbelian} achieves good parameters, but we need to step further here to obtain a growing prime divisor of the full code dimension in Theorem \ref{theorem:intrinsic-code}. To this end, we replace $C$ by a modified code $E$ and define another invariant form $\widetilde{I}_\xi$ in Eqs.~\eqref{eq:modified_complex} and~\eqref{eq:modified_form}, while keeping $A,B$ unchanged. These changes will facilitate us to partition all compatible logical representatives into orbits of the same prime size, as explained in the next subsection. The gauging process can be defined as before with dressed $X$-stabilizers $\widetilde{\mathcal{A}}^{(a)}_{\alpha v}$ for $a = A,B,E$. 
For the family member indexed by $\nu$, the gauged code has the following Hamiltonian:
\begin{align}
	\begin{aligned}
		H_\nu={}&H_Z+\sum_{a\in\{A,B,E\}}\sum_{v_a}
		\left(I-\frac1q\sum_{\alpha\in\mbb F_q}\widetilde{\mathcal{A}}^{(a)}_{\alpha v_a} \right), 
		\qquad \mc C_\nu=\ker H_\nu.
	\end{aligned}
\end{align}
Here $H_Z$ is the sum of the $Z$-check projectors, and $v_a$ range over the coordinate basis vectors indexing the local $X$ generators in $A,B$ and $E$, respectively. Every summand in the Hamiltonian is an orthogonal projector. The original $A,B,C$ are defined on sparse cubical complexes with fixed local coefficient spaces (see Appendix~\ref{appendix:gauged_qLTC}), hence they are qLDPC codes. The modifications on $E$ only enlarge supports by fixed factors, and gauging process uses $I_\xi$ defined cup product also has constant degree and incidence. As a result, $\mc C_\nu$ is a qLDPC code: each term acts on a bounded number of sites and each site meets a bounded number of terms, uniformly in $\nu$. Each site comprises $s_q$ qubits. Write $N_\nu$ for the total number of qubits and $m_\nu = \Theta(N_\nu)$ for the number of projector terms.

\subsection{Growing prime divisors of the ground space dimension}

It is possible to find an explicit orthonormal basis for the gauged code by using the gauge conditions in \cite{LWHL2609nonAbelian}. With the modified code $E$, these conditions read 
\begin{align}\label{eq:gauge_condition_main}
	\begin{aligned}
		\widetilde I_\xi(u,x_B,x_E) & = 0 && \text{for every cocycle } u \in C^{j_A-1}(X,\mathcal F),\\
		\widetilde I_\xi(x_A,v,x_E) & = 0 && \text{for every cocycle } v \in C^{j_B-1}(X,\mathcal G),\\
		\widetilde I_\xi(x_A,x_B,w) & = \lambda(w) && \text{for every cocycle } w \in E^2.
	\end{aligned}
\end{align}
It says the cocycles $x_A,x_B,x_E$ can represent a gauged code states if their cup product and pairing vanish as above. The formal definitions are presented in Section \ref{appendix:gauged_modify}. For now, let $L_\lambda$ be the set of all triples satisfying the three compatibility conditions. Since these states form a basis of the full code,
\begin{align}
	K_\nu \coloneqq \dim\mc C_\nu = |L_\lambda|.
\end{align}
To count $|L_\lambda|$, Ref.~\cite{LWHL2609nonAbelian} simplifies the task by exploring symmetries on the code space. Roughly speaking, the existing construction of (almost) good codes are all covering spaces of some base spaces \cite{LSWLL2026Theory,LLL2026nontrivial,li2026transversalnoncliffordgatesgood}. A cover has several copies of each base cell, with the same local incidence structure and local codes. A deck transformation permutes these copies while preserving their incidences and the base cell below each copy. These transformations therefore induce permutations of the logical cohomology classes. The almost-good family in Ref.~\cite{Dinur2024sheaf} uses a cyclic deck group $C_l$, whereas non-Abelian groups appear in the constructions of good qLDPC codes and qLTCs in Refs.~\cite{PK2022Good,DHLV2022,gay2026asymptoticallygoodquantumlocally,BLN2026}.

The code $E$ provides a chosen deck transformation that acts on the logical states of $A$ and $B$ while leaving the action on $E$ unchanged. The modified invariant form and the above gauge conditions further exclude fixed or invariant triples of gauged code states. 
All other triple lies in an orbit of $p_\nu$ elements where $p_\nu \ge 2\nu+1$ is a prime. As a result,
\begin{align}
	p_\nu\mid K_\nu,\qquad p_\nu\ge2\nu+1\longrightarrow\infty.
\end{align}
and the gauged code dimension has a prime factor tending to infinity. We provide the detailed construct in Appendix~\ref{appendix:gauged_modify}.

\subsection{Operator soundness}\label{subsection:gauged-operator-bound}

We demonstrate in Appendix~\ref{appendix:gauged-parameters} that our gauged code family admits the constant rate and good distance. To apply Theorem~\ref{theorem:intrinsic} to low-energy states of \(H_\nu\), we prove here that the gauged codes has constant \emph{operator soundness}:
\begin{equation}
\Delta_{\mc C_\nu}\leq A H_\nu
\label{equation:section4-target-operator-bound}
\end{equation}
for a constant \(A>0\) independent of \(\nu\).

We first prove the bound on \(\ker H_Z\). On this subspace,
\(H_\nu=\widetilde H_X\), and the dressed \(X\) stabilizers commute. We therefore prove the bound separately on their
common eigenspaces. Label the common eigenspaces by
\(\gamma=(\gamma_{a,v_a})\), where
\(\gamma_{a,v_a}\in\mbb F_q\), and let \(P_\gamma\) denote the
projector onto the eigenspace satisfying
\begin{equation}
\widetilde{\mathcal A}^{(a)}_{\alpha v_a}P_\gamma
=
(-1)^{
\Tr_{\mbb F_q/\mbb F_2}
(\alpha\gamma_{a,v_a})
}
P_\gamma
\label{equation:section4-dressed-x-eigenvalues}
\end{equation}
for every \(a\), \(v_a\), and \(\alpha\).
Let \(|\gamma|\) denote the number of pairs \((a,v_a)\) for
which \(\gamma_{a,v_a}\neq0\). For each pair \((a,v_a)\), the average
over \(\alpha\in\mbb F_q\) equals one when
\(\gamma_{a,v_a}=0\) and zero otherwise. It follows that
\begin{equation}
\widetilde H_XP_\gamma
=
|\gamma|P_\gamma.
\label{equation:section4-dressed-x-energy}
\end{equation}
Projectors \(P_\gamma\) with nonzero range form a complete orthogonal
decomposition of \(\ker H_Z\). In particular, every term of
\(\widetilde H_X\) vanishes on \(\operatorname{ran}P_0\).
Therefore
\begin{equation}
\operatorname{ran}P_0
=
\ker H_Z\cap\ker\widetilde H_X
=
\mc C_\nu.
\label{equation:section4-zero-eigenspace}
\end{equation}

The parent qLTCs have linear soundness and linear distance. 
Let
\(P_\gamma\neq0\). If a Pauli \(Z\) operator maps
\(\operatorname{ran}P_\gamma\) into \(\mc C_\nu\), the
linear soundness gives a representative \(z\) such that
\begin{equation}
Z(z)\operatorname{ran}P_\gamma
\subseteq
\mc C_\nu,
\qquad
\operatorname{wt}(Z(z))
\leq
\kappa_X|\gamma|.
\label{equation:section4-dressed-x-correction}
\end{equation}
If no such Pauli \(Z\) operator, then by linear distance,
\begin{equation}
|\gamma|
\geq
c_*N_\nu.
\label{equation:section4-uncorrectable-eigenvalue-weight}
\end{equation}

Thus, in both cases, \(\operatorname{ran}P_\gamma\) is contained
in a Pauli error space whose radius is at most a constant multiple of \(|\gamma|\). To combine these two bounds, choose an integer
\begin{equation}
\ell\geq\max\{1,\kappa_X,c_*^{-1}\}
\end{equation}
and define
\begin{equation}
r_\gamma
=
\min\{N_\nu,\ell|\gamma|\}.
\label{equation:section4-eigenspace-radius}
\end{equation}
In the first case, since \(\ell\geq\kappa_X\) and every Pauli
operator has weight at most \(N_\nu\), the representative \(z\)
satisfies \(\operatorname{wt}(Z(z))\leq r_\gamma\). In the
second case, \(\ell\geq c_*^{-1}\) gives
\(N_\nu\leq\ell|\gamma|\), so \(r_\gamma=N_\nu\) and
\(V_{r_\gamma}\) is the full Hilbert space. Hence
\begin{equation}
\operatorname{ran}P_\gamma \subseteq V_{r_\gamma}
\label{equation:section4-eigenspace-inclusion}
\end{equation}
for every common eigenspace.

For each \(r\), the direct sum of the eigenspaces satisfying
\(r_\gamma\leq r\) is contained in \(V_r\). The error spaces \(V_r\) decompose according to the
\(Z\)-check syndrome, so \(\Pi_{V_r}\) preserves
\(\ker H_Z\). Taking orthogonal complements within \(\ker H_Z\) gives
\begin{equation}
\left.(I-\Pi_{V_r})\right|_{\ker H_Z} \leq
\sum_{\gamma:\,r_\gamma>r}P_\gamma.
\label{equation:section4-error-space-complement}
\end{equation}
Summing Eq.~\eqref{equation:section4-error-space-complement}
over \(0\leq r<N_\nu\), each \(P_\gamma\) appears
\(r_\gamma\) times. Using \(r_\gamma\leq\ell|\gamma|\) and
Eq.~\eqref{equation:section4-dressed-x-energy}, we obtain
\begin{equation}
\begin{aligned}
\left.\Delta_{\mc C_\nu}\right|_{\ker H_Z}
&\leq \sum_\gamma r_\gamma P_\gamma
\leq \ell\sum_\gamma|\gamma|P_\gamma = \ell\widetilde H_X \Big|_{\ker H_Z}.
\end{aligned}
\label{equation:section4-zero-syndrome-bound}
\end{equation}

We now extend Eq.~\eqref{equation:section4-zero-syndrome-bound}
to states with a nonzero \(Z\)-check syndrome. Since gauging doesn't change \(Z\) checks, the linear \(Z\)-soundness of parent qLTCs guarantee that every \(Z\) syndrome \(\eta\) has a representative \(r\) satisfying
\begin{equation}
w\coloneqq\operatorname{wt}(X(r))
\leq
\kappa_Z|\eta|.
\label{equation:section4-z-syndrome-correction}
\end{equation}
The Pauli \(X(r)\) maps the entire syndrome-\(\eta\) subspace into \(\ker H_Z\).
Multiplication by \(X(r)\) changes the weight of a Pauli error
by at most \(w\), and hence changes its error radius by at most
\(w\). By the definition of \(\Delta_{\mc C_\nu}\), this gives
\begin{equation}
\Delta_{\mc C_\nu}
\leq wI + X(r)^\dagger\Delta_{\mc C_\nu}X(r).
\label{equation:section4-radius-shift}
\end{equation}

Only the terms of \(\widetilde H_X\) whose supports intersect
\(\operatorname{supp}X(r)\) can change under conjugation. Let
\(g_X\) be a uniform upper bound on the number of terms in
\(\widetilde H_X\) containing any given qubit. Thus at most \(g_Xw\) terms can change. Since each term \(h\) is a projector and hence satisfies
\(0\leq h\leq I\), we have
\begin{equation}
X(r)^\dagger\widetilde H_XX(r) \leq \widetilde H_X+g_XwI.
\label{equation:section4-energy-shift}
\end{equation}
Since \(X(r)\) maps the syndrome-\(\eta\) subspace into
\(\ker H_Z\), Eq.~\eqref{equation:section4-zero-syndrome-bound}
implies, on the syndrome-\(\eta\) subspace,
\begin{equation}
X(r)^\dagger\Delta_{\mc C_\nu}X(r)
\leq
\ell X(r)^\dagger\widetilde H_XX(r).
\label{equation:section4-conjugated-zero-syndrome-bound}
\end{equation}
Combining
Eqs.~\eqref{equation:section4-radius-shift},
\eqref{equation:section4-conjugated-zero-syndrome-bound},
and~\eqref{equation:section4-energy-shift}, and using
\(w\leq\kappa_Z|\eta|\), gives, on the syndrome-\(\eta\) subspace,
\begin{equation}
\begin{aligned}
\Delta_{\mc C_\nu} \leq \ell\widetilde H_X+(1+\ell g_X)wI
\leq \ell\widetilde H_X + \kappa_Z(1+\ell g_X)|\eta|I.
\end{aligned}
\label{equation:section4-syndrome-bound}
\end{equation}
On the syndrome-\(\eta\) subspace, \(H_Z=|\eta|I\). Both
\(\Delta_{\mc C_\nu}\) and \(\widetilde H_X\) preserve the
decomposition into \(Z\)-check syndrome subspaces. Therefore, the preceding bounds hold for all \(\eta\):
\begin{equation}
\Delta_{\mc C_\nu} \leq
\ell\widetilde H_X + \kappa_Z(1+\ell g_X)H_Z \leq A H_\nu,
\qquad A=\max\{\ell,\kappa_Z(1+\ell g_X)\}.
\label{equation:section4-full-operator-bound}
\end{equation}
The constant \(A\) is independent of \(\nu\).

Finally, since \(m_\nu=\Theta(N_\nu)\), there exists a constant \(\sigma_*>0\), independent
of \(\nu\), such that
\begin{equation}
\frac{H_\nu}{m_\nu}
\geq
\sigma_*
\frac{\Delta_{\mc C_\nu}}{N_\nu}.
\label{equation:section4-normalized-operator-soundness}
\end{equation}
This proves the operator soundness condition required in
Theorem~\ref{theorem:intrinsic}.

Together with the prime divisor and bounded locality and degree, these statements prove Theorem~\ref{theorem:intrinsic-code}.

\section{Discussion}\label{section:discussion}

We have established an intrinsic notion of NLTM by employing recent advances in sheaf qLDPC codes, substantially broadening the landscape of classical low-energy witnesses excluded in the NLTS approach to qPCP.
This further deepens the connection between quantum error
correction and robust quantum circuit complexity, extending it beyond
stabilizer codes. 
The key insight is that arithmetic properties of the ground space degeneracy in non-Abelian gauged codes can obstruct shallow preparation from arbitrary stabilizer states, while linear distance and constant soundness extend this obstruction throughout a linear energy window. We achieve the target by constructing nonadditive quantum locally testable codes using non-Abelian gauged sheaf codes, combining linear distance and constant soundness while satisfying the arithmetic condition on degeneracy required for intrinsic NLTM. These results reveal a broader role for quantum error correction
in shaping the computational structure of quantum matter:
it can protect complexity beyond entanglement as a robust
low-energy property that persists at nonzero
temperature.



A promising direction is to establish NLTM through structural
principles beyond our dimension criterion, such as conditions
on the fusion rules and braiding statistics of non-Abelian
excitations. This could clarify when non-Abelian quantum order protects complexity at low energies, linking its algebraic structure to the computational properties of quantum matter.

Moreover, here we have shown that soundness naturally links 
energy and error protection in the intrinsic NLTM criterion.
More broadly, it would be valuable to develop complementary
criteria for intrinsic NLTM based on the structure of the
low-energy space, which potentially point to intrinsic NLTM for a wider range
of local Hamiltonians.

Further progress toward qPCP calls for a more complete understanding of classical low-energy witnesses.
A natural direction is to explore broader notions of classical description and simulation,  seeking new mechanisms to exclude increasingly general classes of classical witnesses at positive energy density.



\section*{Acknowledgments}

This work is supported in part by NSFC under Grant No.~12475023, the Dushi Program, and startup funding from YMSC.
F.W.\ acknowledges support from the Shuimu Tsinghua Scholar Program.

The authors developed the key insights and strategies, building on their prior work. Generative AI tools assisted with completing the proof details and aspects of the presentation. The authors take full responsibility for the manuscript.

\appendix

\section{Lemmas for magic circuit lower bounds}\label{appendix:refinement}


We prove the stabilizer completion and rank divisibility statements used in Sec.~\ref{section:arithmetic-obstruction}, together with the spectral estimates used in Appendix~\ref{appendix:ternary}.

\subsection{Qudit stabilizer completion}\label{app:stabcomp}

\begin{proof}[Proof of Lemma~\ref{lemma:stabcompletion}]
Write the Hilbert space as $\bigotimes_v\mbb{C}^{d_v}$ and let $D\coloneqq\prod_vd_v$.
The cyclic shift and clock operators defined in Sec.~\ref{section:statement} give an orthogonal operator basis consisting of the products $\bigotimes_vX_v^{a_v}Z_v^{b_v}$, with $a_v,b_v\in\mbb{Z}_{d_v}$.
Any two Pauli operators commute up to a scalar phase.
These properties hold for unequal and composite local dimensions.

Let $\mc{S}$ be the stabilizer group of $\ket{S_0}$.
It contains no nonidentity scalar, has order $D$, and satisfies
\begin{equation}
\ketbra{S_0}{S_0}=\frac1{D}\sum_{g\in\mc{S}}g.
\end{equation}
Let $G_0$ be the subgroup generated by the elements of $\mc{S}$ whose support is contained in some $A_i$.
The projector onto their common $+1$ eigenspace and its normalized state are
\begin{equation}
P_0=\frac1{|G_0|}\sum_{g\in G_0}g,\quad \rho_0\coloneqq\frac{P_0}{\operatorname{rank}P_0}=\frac1{D}\sum_{g\in G_0}g.
\end{equation}
Partial trace removes every Pauli term whose support is not contained in the retained region.
Since $G_0\subset\mc{S}$ contains every stabilizer supported within each $A_i$, the states $\rho_0$ and $\ketbra{S_0}{S_0}$ have identical reduced states on these regions.
Consequently,
\begin{equation}
\Tr(H'\rho_0)=\bra{S_0}H'\ket{S_0}.
\end{equation}

Choose a local generating set of $G_0$ and multiply each generator by a scalar phase so that it belongs to the Pauli basis $\{\bigotimes_vX_v^{a_v}Z_v^{b_v}\}$.
Extend these basis operators to a family, maximal under inclusion, of pairwise commuting Pauli basis operators, each supported within some $A_i$.
Such a family exists because the Pauli basis is finite.
The chosen operators commute with $G_0$, so their joint eigenspaces within $\im P_0$ give an orthogonal decomposition of $\im P_0$.
The mean energies of the nonzero eigenspaces, weighted by their dimensions, average to $\Tr(H'\rho_0)$.
Choose an eigenspace with mean energy $E_*\le\bra{S_0}H'\ket{S_0}$ and multiply each chosen operator by the inverse of its eigenvalue on this eigenspace.
For the operators obtained from the original generators of $G_0$, this restores their original phases.
Thus the rephased operators generate an abelian group $G$ containing $G_0$, and their common $+1$ eigenspace is exactly the chosen eigenspace.
No nonidentity scalar belongs to $G$, and $G$ is finite.
The projector onto the chosen eigenspace is
\begin{equation}\label{equation:generalized-completion-rank}
P=\frac1{|G|}\sum_{g\in G}g,\quad\text{with } \operatorname{rank}P=\frac{D}{|G|}>0.
\end{equation}
So $\operatorname{rank}P$ is an integer divisor of $D$.

Let $W$ be a Pauli operator supported within one $A_i$.
If $W$ has a nontrivial commutation phase with some element of $G$, then $PWP=0$.
Otherwise, a scalar multiple of $W$ belongs to the Pauli basis and commutes with every operator in the chosen family.
By maximality, this basis operator already belongs to the family.
Hence $W$ is a scalar multiple of an element of $G$, and $PWP$ is a scalar multiple of $P$.
Expanding an arbitrary operator supported within $A_i$ in the Pauli basis proves scalar compression for every such operator.
Applying this to each $O_i$ gives $PH'P=E_*P$, where $E_*$ is the selected mean energy.

Finally, $G$ acts only on $\bigcup_iA_i$ and has a nonzero common $+1$ eigenspace there, so group averaging gives $|G|\le\prod_{v\in\bigcup_iA_i}d_v$.
Choose local generators one at a time, skipping any generated by those already chosen.
Each added generator at least doubles the subgroup size, so
\begin{equation}
s\le\log |G|\le\sum_{v\in\bigcup_iA_i}\log d_v.
\end{equation}
\end{proof}

\subsection{Logical protection and spectral counting}\label{appendix:error-space-protection}

\begin{proof}[Proof of Lemma~\ref{lemma:correctable-space-protection}]
Let $P_{\mc{C}}$ project onto $\mc{C}$ and fix a normalized state $\ket{\psi_0}\in\mc{C}$.
Let $\mc{G}_t$ be the span of $E\ket{\psi_0}$ over all Pauli errors $E$ of weight at most $t$, and let $P_{\mc{G}_t}$ project onto $\mc{G}_t$.
For any such errors $E,F$, $2t<d(\mc{C})$ give
\begin{equation}
P_{\mc{C}}E^\dagger FP_{\mc{C}}=\gamma_{EF}P_{\mc{C}},\quad \gamma_{EF}=\bra{\psi_0}E^\dagger F\ket{\psi_0}.
\end{equation}
This identity shows that the rule
\begin{equation}
W_t(\ket{\psi}\otimes E\ket{\psi_0})\coloneqq E\ket{\psi}
\end{equation}
extends linearly to a well-defined isometry $W_t:\mc{C}\otimes\mc{G}_t\longrightarrow V_t$.
Its image contains every $E\ket{\psi}$ with $\ket{\psi}\in\mc{C}$ and hence equals $V_t$.

Let $O$ be supported on at most $r$ qubits, with $2t+r<d(\mc{C})$.
Every operator $E^\dagger OF$ is supported on fewer than $d(\mc{C})$ qubits and is therefore detected by $\mc{C}$.
Thus
\begin{equation}
P_{\mc{C}}E^\dagger OFP_{\mc{C}}=\beta_{EF}P_{\mc{C}},\quad \beta_{EF}=\bra{\psi_0}E^\dagger OF\ket{\psi_0}.
\end{equation}
Define
\begin{equation}
B_O\coloneqq(P_{\mc{G}_t}OP_{\mc{G}_t})\big|_{\mc{G}_t}.
\end{equation}
Since $E\ket{\psi_0},F\ket{\psi_0}\in\mc{G}_t$, the matrix element of $B_O$ between these vectors is $\beta_{EF}$.
Taking matrix elements between the spanning vectors $\ket{\psi}\otimes E\ket{\psi_0}$ therefore gives
\begin{equation}
W_t^\dagger O W_t=I_{\mc{C}}\otimes B_O.
\end{equation}
The same isometry $W_t$ applies to every such $O$, proving Eq.~\eqref{equation:correctable-space-compression}.
\end{proof}

We now prove the rank divisibility statement using polynomial approximation, logical protection, and the spectral rank obstruction established below.

\begin{proof}[Proof of Lemma~\ref{lemma:intrinsic-rank-divisibility}]
Fix $0<a<1/2$ and set $c_a\coloneqq1+\frac12\ln(2/a)$.
Let $P\coloneqq\Pi_{\ker\widehat{H}}$.
Applying Fact~\ref{fact:stabilizer-filter} to $\widehat{H}$ gives a real polynomial $f$ such that $F\coloneqq f(\widehat{H})$ satisfies
\begin{equation}
\|F-P\|_\infty\le a,\quad \deg f\le L_s(a)\le c_a\sqrt{\max\{1,s\}}.
\end{equation}
Each term in the expansion of $F$ acts on at most $wL_s(a)$ of the $n$ qubits.
The distance condition in \eqref{equation:intrinsic-rank-conditions} therefore gives
\begin{equation}
2t+wL_s(a)<d(\mc{C}).
\end{equation}

Expand each term of $F$ in an operator basis on $A$ as a sum of operators $O_\mu\otimes A_\mu$, where each $O_\mu$ acts on at most $wL_s(a)$ of the $n$ qubits.
Lemma~\ref{lemma:correctable-space-protection} applies to all these operators with the same isometry $W_t$.
By linearity,
\begin{equation}
(W_t^\dagger\otimes I_A)F(W_t\otimes I_A)=I_{\mc{C}}\otimes B
\end{equation}
for an operator $B$ on $\mc{G}_t\otimes A$.
Since $F$ is Hermitian, $B$ is Hermitian.
Recall that $Q=\Pi_{V_t}\otimes I_A$ is the orthogonal projector onto $V_t\otimes A$.
The map $W_t\otimes I_A$ preserves inner products and maps $\mc{C}\otimes\mc{G}_t\otimes A$ onto $\im Q=V_t\otimes A$.
It is therefore a unitary map between these two spaces.
Since $Q$ acts as the identity on the image of $W_t\otimes I_A$, the preceding identity gives
\begin{equation}
(W_t^\dagger\otimes I_A)QFQ(W_t\otimes I_A)=I_{\mc{C}}\otimes B.
\end{equation}
Thus the compression $QFQ\big|_{\im Q}$ is unitarily equivalent to $I_{\mc{C}}\otimes B$.

If $\dim\mc{C}$ did not divide $\operatorname{rank}P$, Lemma~\ref{lemma:spectral-rank-obstruction} below would imply
\begin{equation}
\|P(I-Q)P\|_\infty\ge1-2a,
\end{equation}
contrary to the overlap condition in \eqref{equation:intrinsic-rank-conditions}.
Hence $\dim\mc{C}\mid\operatorname{rank}P=\dim\ker\widehat{H}$.
\end{proof}

\begin{lemma}[Spectral rank obstruction]\label{lemma:spectral-rank-obstruction}
Let $P\ne0$ and $Q$ be orthogonal projectors on a finite dimensional Hilbert space, and let $F$ be Hermitian with $\|F-P\|_\infty\le a$ for some $0<a<1/2$.
Suppose that $QFQ\big|_{\im Q}$ is unitarily equivalent to $I_K\otimes B$ for a Hermitian operator $B$ and a positive integer $K$ that does not divide $\operatorname{rank}P$.
Then
\begin{equation}
\|P(I-Q)P\|_\infty\ge1-2a.
\end{equation}
\end{lemma}

\begin{proof}
Suppose for contradiction that $\|P(I-Q)P\|_\infty<1-2a$.
On $\im P$,
\begin{equation}
PQP=I-P(I-Q)P>2aI.
\end{equation}
Let $T\coloneqq Q\big|_{\im P}$, viewed as a map from $\im P$ to $\im Q$.
The operators $PQP\big|_{\im P}$ and $QPQ\big|_{\im Q}$ are respectively $T^\dagger T$ and $TT^\dagger$, so they have the same nonzero eigenvalues with multiplicities.
The preceding inequality makes $T$ injective.
Thus $QPQ\big|_{\im Q}$ has exactly $\operatorname{rank}P$ positive eigenvalues, all greater than $2a$, and every remaining eigenvalue is zero.

Compression preserves the approximation bound,
\begin{equation}
\|QFQ-QPQ\|_\infty\le a.
\end{equation}
The min--max principle therefore places exactly $\operatorname{rank}P$ eigenvalues of $QFQ\big|_{\im Q}$ in $(a,\infty)$.
The assumed tensor form gives
\begin{equation}
\operatorname{rank}P=\operatorname{rank}\mathbf1_{(a,\infty)}(QFQ\big|_{\im Q})=K\operatorname{rank}\mathbf1_{(a,\infty)}(B),
\end{equation}
contradicting $K\nmid\operatorname{rank}P$.
Hence $\|P(I-Q)P\|_\infty\ge1-2a$.
\end{proof}

\subsection{Chebyshev approximation}

We approximate the projector onto the common satisfying subspace by a polynomial of low degree.
The degree controls the number of constraint supports that occur in each term of the expansion.

\begin{fact}\label{fact:stabilizer-filter}
Let $\Pi_1,\cdots,\Pi_s$ be commuting orthogonal projectors, let $H_P\coloneqq\sum_{j=1}^s(I-\Pi_j)$, and let $P$ project onto $\ker H_P$.
For $0<a<1/2$, there is a real polynomial $f$ of degree at most
\begin{equation}
L_s(a)\coloneqq\left\lceil\frac{\sqrt{s}}2\ln\frac2a\right\rceil
\end{equation}
with $f(0)=1$ and $\|f(H_P)-P\|_\infty\le a$.
Moreover, $f(H_P)$ is a linear combination of products of at most $L_s(a)$ of the projectors $\Pi_j$.
\end{fact}

\begin{proof}
The commuting projector Hamiltonian $H_P$ has spectrum in $\{0,1,\cdots,s\}$.
For $s\ge2$, use the Chebyshev spectral filter~\cite{arad2013arealawsubexponentialalgorithm} with $L\coloneqq L_s(a)$,
\begin{equation}
f(x)\coloneqq\frac{T_L((s+1-2x)/(s-1))}{T_L((s+1)/(s-1))},
\end{equation}
where $T_L(\cos\theta)=\cos(L\theta)$ defines the degree-$L$ Chebyshev polynomial.
We have $f(0)=1$, and the numerator has absolute value at most one for $x\in[1,s]$.
To bound the denominator, note that
\begin{equation}
\operatorname{arcosh}\frac{s+1}{s-1}=2\operatorname{artanh}\frac1{\sqrt{s}}\ge\frac2{\sqrt{s}}.
\end{equation}
Using $T_L(\cosh u)=\cosh(Lu)\ge e^{Lu}/2$, we obtain
\begin{equation}
\max_{1\le j\le s}|f(j)|\le2e^{-2L/\sqrt{s}}\le a.
\end{equation}
Applying this estimate to every eigenspace gives $\|f(H_P)-P\|_\infty\le a$.

Expanding a polynomial of degree at most $L_s(a)$ in $H_P$ gives products of at most $L_s(a)$ of the projectors $\Pi_j$.
The support of each product is contained in the union of the supports of its factors.
\end{proof}

\section{Construction and properties of gauged sheaf codes}\label{appendix:gauged_code}

\subsection{Introduction to sheaf codes and cup product gauging}\label{appendix:gauged_intro}

For further constructions, we introduce the sheaf cochain complex and cup product gauging. More details can be found in, e.g., \cite{DHLV2022,Panteleev2024,Lin2024transversal,christos2026nonabelianquantumlowdensityparity,Zhu2026,zhu2026nonAbelian,Li2025Poincare,LSWLL2026Theory,LLL2026nontrivial,LWHL2609nonAbelian}. Let $X$ be a finite \emph{cubical complex} and let $X(j)$ denote its $j$-cells or cubes. A \emph{sheaf} $\mathcal{F}$ assigns a finite-dimensional vector space $\mathcal{F}_\sigma$ over $\mbb{F}_q$ with $q=2^{s_q}$ and $s_q$ a positive integer to each cell $\sigma$, and a linear map $\mathcal{F}_{\sigma,\tau}: \mathcal{F}_\sigma \to \mathcal{F}_\tau$ to each inclusion $\sigma \leq \tau$. These maps satisfy $\mathcal{F}_{\tau,\upsilon} \mathcal{F}_{\sigma,\tau} = \mathcal{F}_{\sigma,\upsilon}$ and $\mathcal{F}_{\sigma,\sigma} = \mathrm{id}$. The cochain spaces and coboundary maps are
\begin{align}\label{eq:sheaf_cochain}
	C^j(X,\mathcal F) = \bigoplus_{\sigma \in X(j)} \mathcal F_\sigma,
	\qquad
	(\delta^j x)(\tau) = \sum_{\substack{\sigma\in X(j) \\ \sigma \prec \tau}}
	\mathcal F_{\sigma,\tau}x(\sigma), \quad \tau \in X(j+1),
\end{align}
where the sum runs over all $j$-cells $\sigma$ contained in $\tau$, and $x(\sigma) \in \mathcal F_\sigma$ is its \emph{local coefficient vector}. After choosing bases, the matrix block of $\delta^j$ from $\sigma$ to $\tau$ is $\mathcal{F}_{\sigma,\tau}$ when $\sigma\prec\tau$, and zero otherwise. 
A sheaf code uses three consecutive terms
\begin{align}\label{eq:sheaf_code}
	C^{j-1}(X,\mathcal{F}) \xrightarrow{\delta^{j-1}}
	C^j(X,\mathcal{F}) \xrightarrow{\delta^j} C^{j+1}(X,\mathcal{F}).
\end{align}
Each coordinate of $C^j(X,\mathcal{F})$ stands for one physical $q$-dimensional site. 
The dual spaces $C_j(X,\mathcal F) \coloneq C^j(X,\mathcal F)^*$ with \emph{boundary operator} $\partial_j = (\delta^{j-1})^T$ form the chain complex. 
For $y \in C^j(X,\mathcal{F})$ and $z \in C_j(X,\mathcal{F}) \cong C^j(X,\mathcal{F})$, the Pauli operators are defined by the trace formula (cf. the generalized Pauli operators used in Section \ref{section:mgaic_complexity}):
\begin{align}
	X(y)\ket{x} \coloneq \ket{x+y}, \qquad Z(z)\ket{x} \coloneq (-1)^{\Tr_{\mbb F_q/\mbb F_2} \langle z,x \rangle} \ket{x}.
\end{align}
The $Z$ checks require $\delta^px=0$, and the $X$ checks identify cochains differing by a coboundary. A cohomology class and the corresponding logical quantum state are
\begin{align}
	[x] \in H^p(X,\mathcal F)=\ker\delta^p/ \im \delta^{p-1}, \quad
	\ket{[x]} = \frac{1}{|\im \delta^{p-1}|^{1/2}} \sum_{b \in \im \delta^{p-1}} \ket{x+b}.
\end{align}

The gauged code is defined by using the cellular cup product with sheaves~\cite{Li2025Poincare,LLL2026nontrivial}: 
\begin{align}
	\smile: C^j(X,\mathcal{F}) \times C^k(X,\mathcal{G}) \longrightarrow C^{j+k}(X,\mathcal{F} \otimes \mathcal{G}).
\end{align}
The coefficient spaces and maps of $\mathcal{F} \otimes \mathcal{G}$ are the corresponding tensor products. The cup product combines cochain values on complementary faces of a cube. For example, on a two-dimensional cube, the product of two 1-cochains sums over the two pairs of complementary initial and terminal edges: evaluate the first cochain on the initial edge and the second on the terminal edge, restrict both local coefficient vectors to the square, and take their tensor product. 
Iterating the cup product gives a product of more than two cochains. In the following construction, we need three sheaves $\mathcal{F},\mathcal{G},\mathcal{H}$, and a nonzero cycle $\xi \in C_{j_A + j_B + j_C}(X,\mathcal F\otimes\mathcal G\otimes\mathcal H)$, so $\partial\xi = 0$. For cochains $x_A,x_B,x_C$ with coefficients in $\mathcal F,\mathcal G,\mathcal H$, respectively, the invariant form in \cite{li2026transversalnoncliffordgatesgood,LWHL2609nonAbelian} is essential to the constructions:
\begin{align}\label{eq:invariant_form}
	\begin{aligned}
		I_\xi(x_A,x_B,x_C) & = \int_\xi x_A \smile x_B \smile x_C
		= \langle x_A \smile x_B \smile x_C, \ \xi \rangle \\
		& = \sum_{\tau \in X(j_A + j_B + j_C)} \xi(\tau) \cdot \bigl((x_A\smile x_B\smile x_C)(\tau)\bigr),
	\end{aligned}
\end{align}
Here the discrete integral $\int_\xi$ or the inner product can be used interchangeably as pairing between chains and cochains of the same degree. When the other two inputs are cocycles, the Leibniz rule and $\partial \xi = 0$ imply
\begin{align}\label{eq:Leibniz_rule}
	I_\xi (\delta u_A,x_B,x_C)=I_\xi(x_A,\delta u_B,x_C)=I_\xi(x_A,x_B,\delta u_C) = 0.
\end{align}
Consequently the form on cocycles depends only on their cohomology classes. The cubical complex $X$ used throughout the paper has bounded cell incidence and fixed coefficient dimensions, which implies that each coordinate occurs in only a bounded number of its terms. A nonzero induced trilinear form yields constant-depth non-Clifford gates \cite{LLL2026nontrivial,li2026transversalnoncliffordgatesgood}. The same construction also gives gauged good qLDPC codes \cite{LWHL2609nonAbelian}.

We now define the gauged codes. Label the three blocks by $a=A,B,C$, with sheaves $\mathcal F,\mathcal G,\mathcal H$, respectively, and write $\delta_A,\delta_B,\delta_C$ for their coboundaries, omitting the degrees. Choose physical degrees $j_A,j_B,j_C \ge 1$ satisfying $j_A + j_B + j_C - 1$ equal to the top degree $t$. For $x_A \in C^{j_A}(X,\mathcal F)$, $x_B \in C^{j_B}(X,\mathcal G)$ and $x_C \in C^{j_C}(X,\mathcal H)$, we consider the dressed $X$ stabilizer generators
\begin{equation}\label{eq:dressed_X}
	\begin{aligned}
		\mathcal{A}^{(A)}_{\alpha_A v_A} \ket{x_A,x_B,x_C}
		& = (-1)^{\Tr_{\mbb F_q/\mbb F_2} I_\xi(\alpha_A v_A, x_B,x_C)}
		\ket{x_A + \alpha_A \delta_A v_A ,x_B,x_C}, \\
		\mathcal{A}^{(B)}_{\alpha_B v_B} \ket{x_A,x_B,x_C}
		& = (-1)^{\Tr_{\mbb F_q/\mbb F_2} I_\xi(x_A, \alpha_B v_B,x_C)} \ket{x_A,x_B + \alpha_B \delta_B v_B, x_C}, \\
		\mathcal{A}^{(C)}_{\alpha_C v_C} \ket{x_A,x_B,x_C}
		& = (-1)^{\Tr_{\mbb F_q/\mbb F_2} I_\xi(x_A,x_B,\alpha_C v_C)} \ket{x_A,x_B,x_C + \alpha_C \delta_C v_C}.
	\end{aligned}
\end{equation}
where $\alpha_a\in\mbb F_q$ for $a=A,B,C$, and $v_A,v_B,v_C$ are coordinate basis vectors of $C^{j_A-1}(X,\mathcal F)$, $C^{j_B-1}(X,\mathcal G)$ and $C^{j_C-1}(X,\mathcal H)$, respectively. The phase in Eq.~\eqref{eq:dressed_X} are realized by physical multi-controlled-$Z$ gates of constant weight. On the other hand, the $Z$ stabilizer generators are unchanged and can be read off from the coboundary operators.

The gauged code $\mathcal C_G$ is the common $+1$ eigenspace of these stabilizers: $\mathcal{C}_G = \ker H_G$, where, as operators on quantum states, the local terms and the Hamiltonian are
\begin{align}
	h^Z_{a,r}\ket{x_A,x_B,x_C} & = \mathbf1_{(\delta_a x_a)_r\ne0}\ket{x_A,x_B,x_C},
	\qquad H_Z = \sum_{a,r}h^Z_{a,r}, \label{eq:undressed_Z} \\
	h^X_{a,v} & = I-\frac1q\sum_{\alpha\in\mbb F_q}\mathcal A^{(a)}_{\alpha v},
	\qquad\qquad H_G = H_Z+\sum_{a,v}h^X_{a,v}. \label{eq:gauged_hamiltonian}
\end{align}
Here $a=A,B,C$, $r$ indexes the coordinates of $\delta_a x_a$, and $v$ runs over the coordinate basis for the $a$-th family in Eq.~\eqref{eq:dressed_X}.

Let $L_G \subseteq H^{j_A}(X,\mathcal{F}) \times H^{j_B}(X,\mathcal{G}) \times H^{j_C}(X,\mathcal{H})$ consist of the triples $([x_A],[x_B],[x_C])$ satisfying
\begin{align}\label{equation:general-gauged-classes}
	\begin{aligned}
		I_\xi(u,x_B,x_C)&=0&&\text{for every cocycle }u\in C^{j_A-1}(X,\mathcal F),\\
		I_\xi(x_A,v,x_C)&=0&&\text{for every cocycle }v\in C^{j_B-1}(X,\mathcal G),\\
		I_\xi(x_A,x_B,w)&=0&&\text{for every cocycle }w\in C^{j_C-1}(X,\mathcal H).
	\end{aligned}
\end{align}
It is proved in Ref.~\cite{LWHL2609nonAbelian} that each triple satisfying the above gauge conditions defines one orthonormal basis state of $\mathcal{C}_G$. In particular, $\dim \mathcal{C}_G = |L_G|$. We will apply this fact later in our setting.

\subsection{Gauging good quantum locally testable codes}\label{appendix:gauged_qLTC}
 
Now, we need three code blocks $A,B,C$ with linear distance and constant soundness, and the first of which has positive rate. These ingredients are given in the recent construction of good qLTCs \cite{gay2026asymptoticallygoodquantumlocally,BLN2026,CHLT2026}. The additional choices in \cite{li2026transversalnoncliffordgatesgood} give a nonzero trilinear invariant form. We then modify the third block to ensure that the gauge conditions are preserved by an action of growing prime order. The resulting complex is denoted by $E$ and this gives a prime divisor of the gauged code dimension tending to infinity, as claimed in Theorem~\ref{theorem:intrinsic} and \ref{theorem:intrinsic-code}. The parameters of $E$ and of the gauged code are proved in Appendix~\ref{appendix:gauged-parameters}.

We first outline construction of good qLTCs: write $s_0=2^{h_s}$, where $h_s$ is a sufficiently large positive integer with $3\nmid h_s$. Let $D_1,\ldots,D_6 \geq 3$ be the sufficiently large distinct odd degrees and let
\begin{align}
	Q_i = s_0^{2D_i}=2^{2h_sD_i},\qquad n_i=Q_i+1,\qquad 1\le i\le6.
\end{align}
All these numbers remain constant as the code length grows. The distinct $D_i$ make the $Q_i$ distinct, as required to obtain product-expanding projective Reed--Solomon in \cite{gay2026asymptoticallygoodquantumlocally}.

We define a $6$-dimensional cubical complex $X$. For $1\le i\le6$, let $T_i$ be the Bruhat--Tits tree for a completion of $\mbb F_{s_0^2}(z)$ at a degree-$D_i$ place. Its residue field is $\mbb F_{Q_i}$, so $T_i$ has degree $n_i$. The quaternionic construction used in \cite{li2026transversalnoncliffordgatesgood} supplies a fixed arithmetic group $\Gamma_0$ acting freely and cocompactly from the left on $T_1\times\cdots\times T_6$. Each tree is bipartite, with vertex types $0$ and $1$, and the action preserves the directions and these types. The quotient complex defined by the left action $X_0 = \Gamma_0 \backslash (T_1\times\cdots\times T_6)$ is finite.

Following Ref.~\cite{li2026transversalnoncliffordgatesgood}, we let $\nu$ range over an unbounded sequence of odd positive integers. For each such $\nu$, we choose four distinct irreducible polynomials $w_{\nu,1},\ldots,w_{\nu,4}\in\mbb F_{s_0^2}[z]$, all of degree $\nu$, so that the invariant form $I_\xi$ on the resulting cover is nonzero.
Reduction modulo these polynomials gives a homomorphism
\begin{align}\label{eq:arithmetic-cover}
	\rho_\nu:\Gamma_0\longrightarrow G_\nu\coloneq\mathrm{SL}_2(\mbb F_{Q_\nu})^4,
	\qquad Q_\nu=s_0^{2\nu},\qquad \Gamma_\nu=\ker\rho_\nu.
\end{align}
Each residue field $\mbb F_{s_0^2}[z]/(w_{\nu,j})$ has size $Q_\nu$. Strong approximation for simply connected groups over function fields~\cite{Prasad1977StrongApproximation} gives simultaneous surjectivity of these reductions. This yields the normal covering
\begin{align}
	X=X_\nu=\Gamma_\nu\backslash(T_1\times\cdots\times T_6)
	\xrightarrow{P_\nu}X_0,
	\qquad v_\nu=|G_\nu|=[Q_\nu(Q_\nu^2-1)]^4
\end{align}
with $v_\nu$ sheets and non-Abelian deck group $G_\nu$. The covering map is $P_\nu(\Gamma_\nu \tau) \coloneq \Gamma_0 \tau$ for a cell $\tau$ of the product of trees $T_1\times\cdots\times T_6$.

Over the covering space, the desired three sheaves are defined by pullbacks $\mathcal F=P_\nu^*\mathcal F_0$, $\mathcal G = P_\nu^*\mathcal G_0$ and $\mathcal H=P_\nu^*\mathcal H_0$, where $\mathcal F_0,\mathcal G_0,\mathcal H_0$ are the three sheaves on $X_0$ generated by projective Reed--Solomon codes~\cite{gay2026asymptoticallygoodquantumlocally,li2026transversalnoncliffordgatesgood}. We skip the formal definitions here, but the basic intuitions is that the local coefficient spaces and restriction matrices at a lifted cell or incidence are copied from its image in $X_0$. In direction $i$, the local codes are projective Reed--Solomon codes on $\mbb P^1(\mbb F_{Q_i})$ of length $n_i$. To use the same physical local dimension in every direction of $X$, we choose one finite field $\mbb F_q$ containing all fixed coefficients of the base sheaves and the fields $\mbb F_{Q_i^2}$, $1\le i\le6$. Such a common finite extension exists. We will regard all local matrices as matrices over $\mbb F_q$ and take their row spans over that field, which remains fixed throughout the covering family. The three sheaf codes are defined by $C^1(X,\mathcal F)\to C^2(X,\mathcal F)\to C^3(X,\mathcal F)$ and the corresponding terms for $\mathcal G$ and $\mathcal H$, as in Eq.~\eqref{eq:sheaf_code}. Since $X$ is defined with bounded cell incidence and fixed local coefficient dimensions, the sheaf codes defined here are qLDPC codes.

Now the index $a=A,B,C$ labels the three code blocks, and $i=1,\ldots,6$ labels their geometric directions in $X$. Let $B_A=\{1,2\}$, $B_B=\{3,4\}$, $B_C=\{5,6\}$, and let $e_i = (Q_i-1)/(s_0+1)$. In block $a$ and direction $i$, the local code is required to have dimension
\begin{align}\label{eq:local_rank}
	m_{a,i}=1+\begin{cases}
		e_i,&i\in B_a,\\
		(s_0-9)e_i,&a=A,\ i\notin B_A,\\
		3e_i,&a=B,C,\ i\notin B_a.
	\end{cases}
\end{align}
These local dimensions support the existence of a nonzero invariant form $I_\xi$ constructed in \cite{li2026transversalnoncliffordgatesgood}. For sufficiently large fixed $s_0$, the first code block has constant rate. The Ramanujan expansion and local product expansion give these three codes linear distance and constant soundness \cite{gay2026asymptoticallygoodquantumlocally}. 

We now consider the deck transformation $s\in G_\nu$ on cells $\sigma \prec \tau$ of $X$. We will use a cyclic subgroup of this action to prove prime divisibility of the code dimension, as in the group action counting argument of \cite{LWHL2609nonAbelian}. The covering and pullback definitions give
\begin{align}\label{eq:deck_coboundary}
		& P_\nu(s\sigma) = P_\nu\sigma,\qquad s\sigma\prec s\tau,
		\qquad \mathcal{F}_{s\sigma,s\tau} = \mathcal{F}_{\sigma,\tau},\\
		& (sx)(\sigma)=x(s^{-1}\sigma),\qquad \delta_a(sx)=s(\delta_a x).
\end{align}
The same restriction identity holds for $\mathcal G$ and $\mathcal H$, and $\delta_a$ denotes the coboundary in block $a=A,B,C$. Commutation with the coboundary follows from the covering identities, without requiring an abelian group. Thus $s[x] \coloneq [sx]$ is a well-defined action on cohomology.

\subsection{The modified complex and growing prime divisors}\label{appendix:gauged_modify}

Given $C^\bullet(X,\mathcal F)$, $C^\bullet(X,\mathcal G)$ and $C^\bullet(X,\mathcal H)$, we choose $\xi$ as in \cite{li2026transversalnoncliffordgatesgood} by copying the local coefficients of a nonzero 6-cycle in the base $X_0$ to every lifted cell. This cycle is fixed by deck transformations, so $I_\xi$ in Eq.~\eqref{eq:invariant_form} is also invariant under their simultaneous action on three inputs. We will keep the first two complexes unchanged and modify the third so that its contribution to the gauge conditions comes from cohomology classes fixed by an element $\rho \in G_\nu$ of growing prime order.

For every sufficiently large $\nu$ in the chosen sequence, there is a primitive prime divisor $p = p_\nu$ of $s_0^{2\nu}-1$~\cite{Ribenboim2004Primes}. This means
\begin{align}
	p \mid s_0^{2\nu}-1,\qquad p \nmid s_0^k-1 \quad (1\le k<2\nu).
\end{align}
Its multiplicative order is $\mathrm{ord}_p(s_0) \coloneq \min\{k \geq 1:s_0^k \equiv 1 \pmod p\}$. With $e_\nu = (Q_\nu-1)/(s_0+1)$, we therefore have
\begin{align}\label{eq:prime_order}
	\mathrm{ord}_p(s_0)=2\nu,\qquad p\ge2\nu+1,\qquad p\mid e_\nu.
\end{align}
The last divisibility follows because a prime dividing $s_0 + 1$ has order at most two.

Choose $\omega\in\mbb F_{Q_\nu}^{\times}$ of order $p$. Each factor of $G_\nu=\mathrm{SL}_2(\mbb F_{Q_\nu})^4$ consists of the $2\times2$ matrices of determinant one. Set $\rho = (\mathrm{diag}(\omega,\omega^{-1}),I_2,I_2,I_2)$. Its action on cochains permutes local coefficient vectors and is $\mbb F_q$-linear. Note that here $q$ stays fixed while $Q_\nu$ tends to infinity.

We modify the third complex $C^\bullet(X,\mathcal H)$ as follows:
\begin{align}\label{eq:modified_complex}
	& E^j=C^j(X,\mathcal H)\oplus C^{j-1}(X,\mathcal H), \\
	& \delta_E:E^j\longrightarrow E^{j+1},\qquad
	\delta_E(y,x)=(\delta_Cy,y-\rho y+\delta_Cx), \\
	& \pi:E^j\longrightarrow C^j(X,\mathcal H), \qquad \pi(y,x)=y.
\end{align}
Here $\delta_E^j$ is coboundary on $E^j$, with its degree omitted when clear, and $\pi$ is the projection onto the first summand. The formula defines $E$ in every degree from zero through seven, with $E^0 = C^0(X,\mathcal H)$ and $E^7=C^6(X,\mathcal H)$. Since $\rho$ commutes with $\delta_C$ and the field has characteristic two, $\delta_E^{j+1}\delta_E^j=0$ and $\pi \delta_E = \delta_C\pi$. A cocycle $y\in C^j(X,\mathcal H)$ lifts to a cocycle $(y,x)\in E^j$ exactly when $y-\rho y=\delta_Cx$. Consequently,
\begin{align}\label{eq:fiber_invariant_image}
	\operatorname{im}\bigl( \pi_*: H^j(E) \longrightarrow H^j(X,\mathcal H) \bigr)
	= H^j(X,\mathcal H)^\rho,
\end{align}
where $H^j(X,\mathcal H)^\rho \coloneq \{[y]:[\rho y] = [y]\}$. As a result, the projection of every cohomology class of $E$ is fixed by $\rho$. We are going to prove that gauging code blocks $A$, $B$ with $E$ yields a code whose dimension has growing prime divisors. The following lemma is the first step.

\begin{lemma}\label{lemma:seed_cup}
	For the element $\rho$, there are 2-cocycles $a_0 \in C^2(X,\mathcal F)$, $b_0 \in C^2(X,\mathcal G)$ and $c_0 \in C^2(X,\mathcal H)$ satisfying
	\begin{align}\label{eq:arithmetic_nonvanishing}
		[\rho c_0] = [c_0],\qquad I_\xi(a_0,b_0 - \rho b_0,c_0) \neq 0.
	\end{align}
\end{lemma}
\begin{proof}
	We use three formal polynomials in $X,Y$ obtained from the construction of Refs.~\cite{li2026transversalnoncliffordgatesgood} over a finite extension $\mathbb K$ containing $\mbb F_q$, $\mbb F_{Q_\nu}$. The polynomials have the form
	\begin{align}
		\mathcal P_A(X,Y) & = \sum_F a_F F(X,Y),
		\qquad a_F \in C^2(X,\mathcal F)\otimes_{\mbb F_q} \mathbb K,\\
		\mathcal P_B(X,Y) & = \sum_{F,G} b_{F,G}H_{F,G}(X,Y),
		\qquad b_{F,G} \in C^2(X,\mathcal G)\otimes_{\mbb F_q}\mathbb K,\\
		\mathcal P_C(X,Y) & = \sum_G c_G G(X,Y),
		\qquad c_G \in C^2(X,\mathcal H)\otimes_{\mbb F_q}\mathbb K.
	\end{align}
	Here $F$ ranges over the monomials with nonzero coefficients in $(X+Y)^{s_0e_\nu}$, and $G$ ranges over those in $(X+Y)^{e_\nu}$, with both expansions taken over $\mathbb K$. Their products form a basis of the $\mathbb K$-valued functions on $\mbb P^1(\mbb F_{Q_\nu})$ whose sum is zero. The polynomial $H_{F,G}$ is the homogeneous polynomial of degree $Q_\nu-1$ dual to $FG$: for any pair $F',G'$ from these bases,
	\begin{align}
		\sum_{z\in\mbb P^1(\mbb F_{Q_\nu})}H_{F,G}(z)(F'G')(z)
		= \begin{cases}1,&F'=F,\ G'=G,\\0,&\text{otherwise}.\end{cases}
	\end{align}
	Every coefficient $a_F,b_{F,G},c_G$ is a 2-cocycle. Indeed, applying the corresponding coboundary to each polynomial gives zero, and the displayed basis polynomials are linearly independent.The polynomials above are obtained by selecting a monomial pair and its dual function in each of the last three factors of $G_\nu$ and extracting the corresponding coefficients. These coefficients satisfy $I_\xi(a_F,b_{F,G},c_G)\neq0$ for every pair $F,G$~\cite{li2026transversalnoncliffordgatesgood}. Since $\rho$ is the identity in those three group factors, these choices do not affect its action on the remaining variables $X,Y$.
	
	Choose
	\begin{align}
		f=X^{s_0e_\nu-s_0}Y^{s_0},\qquad g = X^{e_\nu},\qquad h=X^{s_0}Y^{Q_\nu-1-s_0}.
	\end{align}
	The monomial $g$ occurs in $(X+Y)^{e_\nu}$, and $f$ occurs in $(X+Y)^{s_0e_\nu}=(X^{s_0}+Y^{s_0})^{e_\nu}$ because $e_\nu$ is odd. Regard $h$ as a function on the projective line. For every allowed monomial pair $u,v$ of degrees $s_0 e_\nu, e_\nu$, finite field power sums give
	\begin{align}
		\sum_{z\in\mbb P^1(\mbb F_{Q_\nu})}h(z)(uv)(z) =
		\begin{cases}
			1,&u=f,\ v=g,\\
			0,&\text{otherwise}.
		\end{cases}
	\end{align}
	Thus $h = H_{f,g}$. Both $h$ and $uv$ have degree $Q_\nu-1$, so their evaluations are independent of the chosen projective representatives. Take $a=a_f$, $b=b_{f,g}$ and $c = c_g$. On these polynomials, $\rho$ acts by $u(X,Y)\mapsto u(\omega^{-1}X,\omega Y)$. Since $p\mid e_\nu$,
	\begin{align}
		\rho f=\omega^{2s_0}f,\qquad \rho h=\omega^{-2s_0}h,\qquad \rho g=g.
	\end{align}
	Applying $\rho$ to the cohomology class of each coefficient gives
	\begin{align}
		\rho[\mathcal P_D(X,Y)]=[\mathcal P_D(\omega X,\omega^{-1}Y)],\qquad D=A,B,C,
	\end{align}
	where brackets are applied coefficientwise. Comparing the coefficients of $f,h,g$ gives
	\begin{align}
		\rho[a]=\omega^{-2s_0}[a],\qquad \rho[b]=\omega^{2s_0}[b],\qquad \rho[c]=[c].
	\end{align}
	For cocycle representatives, this gives
	\begin{align}
		I_\xi(a,b-\rho b,c)=(1-\omega^{2s_0})I_\xi(a,b,c)\ne0,
	\end{align}
	because $p$ is odd and $s_0$ is a power of two. It remains to choose the cocycles over $\mbb F_q$. The trilinear form and the subspaces $\im (1-\rho)\subseteq H^2(X,\mathcal G)$ and $H^2(X,\mathcal H)^\rho$ are defined over $\mbb F_q$. Kernels and images commute with field extension. If the form vanished on $H^2(X,\mathcal F)\times \im(1-\rho)\times H^2(X,\mathcal H)^\rho$ over $\mbb F_q$, it would also vanish after extension to $\mathbb K$, contrary to the displayed calculation. We can therefore choose such classes over $\mbb F_q$, take a preimage under $1-\rho$ in the second cohomology, and choose cocycle representatives $a_0,b_0,c_0$ as required.
\end{proof}

With the above result, we now define
\begin{align}\label{eq:modified_form}
	\widetilde{I}_\xi(a,b,e) \coloneq I_\xi(a,b - \rho b,\pi e).
\end{align}
Both $1 - \rho$ and $\pi$ commute with coboundaries, so $\widetilde{I}_\xi$ satisfies the Leibniz rule. By Lemma~\ref{lemma:seed_cup}, there is $u_0 \in C^1(X,\mathcal H)$ with $c_0- \rho c_0 = \delta_C u_0$. Hence $w_0 = (c_0,u_0) \in E^2$ is a cocycle and $\widetilde{I}_\xi (a_0,b_0,w_0) \neq 0$. Let
\begin{align}\label{eq::E_lamabda}
	\lambda: E^2 \longrightarrow \mbb{F}_q,\qquad \lambda(w) = \widetilde{I}_\xi(a_0,b_0,w).
\end{align}
This functional vanishes on $\delta_E E^1$ and is nonzero on $w_0$. We also write $\lambda$ for its induced nonzero functional on $H^2(E)$.

We apply the preceding gauging rules to
\begin{align}\label{eq:-complexes}
	C^1(X,\mathcal F)&\xrightarrow{\delta_A}C^2(X,\mathcal F)\xrightarrow{\delta_A}C^3(X,\mathcal F),\\
	C^1(X,\mathcal G)&\xrightarrow{\delta_B}C^2(X,\mathcal G)\xrightarrow{\delta_B}C^3(X,\mathcal G),\\
	E^2&\xrightarrow{\delta_E}E^3\xrightarrow{\delta_E}E^4,
\end{align}
and use $\widetilde I_\xi$ in place of $I_\xi$. These triples specify the physical spaces and stabilizers. The first two families of dressed $X$ stabilizer generators have eigenvalue $+1$. For the third family, prescribe eigenvalue $(-1)^{\Tr_{\mbb F_q/\mbb F_2}\lambda(w)}$ for $w\in E^2$. Equivalently, impose eigenvalue $+1$ on
\begin{align}\label{eq:dressed_E}
	\widetilde{\mathcal A}^{(E)}_w\ket{a,b,e} \coloneq (-1)^{\Tr_{\mbb F_q/\mbb F_2}
		\left(\widetilde I_\xi(a,b,w)-\lambda(w)\right)} \ket{a,b,e+\delta_Ew},
\end{align}
The $Z$ stabilizers are intact as before. In the matrix of $\delta_E$, the new term $1-\rho$ adds at most two nonzero entries to each relevant row and column. In both  $\widetilde{I}_\xi$ and $\lambda$, a term of the original cup product is replaced by its two copies using $b$ and $\rho b$, while $\pi$ simply discards the second summand. Each dressed $X$ stabilizer generator therefore acts on a bounded number of coordinates, and each coordinate occurs in a bounded number of checks. The present gauging process still yields qLDPC codes.

For a cocycle $w \in E^2$, the translation in Eq.~\eqref{eq:dressed_E} is zero, so its prescribed eigenvalue requires $\widetilde I_\xi(a,b,w) = \lambda(w)$. Since all scalar multiples of $w$ are included,  the trace conditions are thus equivalent to this $\mbb F_q$-valued equation. The first two families have eigenvalue $+1$, so their compatibility equations are unchanged. Let $L_\lambda$ be the full set of triples
\begin{align}
	([a],[b],[e])\in H^2(X,\mathcal F)\times H^2(X,\mathcal G)\times H^3(E)
\end{align}
satisfying
\begin{align}\label{eq:gauge_condition_E}
	\widetilde I_\xi(u,b,e) & = 0 \qquad \text{for every cocycle }u\in C^1(X,\mathcal F),\\
	\widetilde I_\xi(a,v,e) & = 0 \qquad \text{for every cocycle }v\in C^1(X,\mathcal G),\\
	\widetilde I_\xi(a,b,w) & = \lambda(w) \qquad \text{for every cocycle }w\in E^2.
\end{align}
Each triple satisfying all three conditions contributes one basis state of the gauged code~\cite{LWHL2609nonAbelian}. Thus
\begin{align}\label{eq:E-full-count}
	K_\nu \coloneq \dim\mathcal C_G = |L_\lambda|.
\end{align}
The triple $([a_0],[b_0],0)$ lies in $L_\lambda$, so $K_\nu>0$. Consider the action $	([a],[b],[e])\longmapsto([\rho a],[\rho b],[e])$. For $[e] \in H^3(E)$ and the cocycles $u,v,w$ in Eq.~\eqref{eq:gauge_condition_E}, Eq.~\eqref{eq:fiber_invariant_image} gives $\rho\pi_*[e]=\pi_*[e]$ and $\rho\pi_*[w]=\pi_*[w]$. Invariance of $I_\xi$ therefore gives
\begin{align}
	\begin{aligned}
		\widetilde I_\xi(u,\rho b,e)&=\widetilde I_\xi(\rho^{-1}u,b,e),\\
		\widetilde I_\xi(\rho a,v,e)&=\widetilde I_\xi(a,\rho^{-1}v,e),\\
		\widetilde I_\xi(\rho a,\rho b,w)&=\widetilde I_\xi(a,b,w).
	\end{aligned}
\end{align}
Since $\rho^{-1}$ permutes the cocycles in $C^1(X,\mathcal F)$ and $C^1(X,\mathcal G)$, the first two identities preserve the first two conditions in Eq.~\eqref{eq:gauge_condition_E}. The third identity preserves its last condition. Thus the action preserves $L_\lambda$. A fixed triple would have $[\rho b] = [b]$, making $b - \rho b$ a coboundary and forcing $\widetilde I_\xi(a,b,w)=0$ for every cocycle $w \in E^2$. This contradicts $\lambda([w_0]) \neq 0$. Since $\rho$ has prime order $p_\nu$, every orbit has $p_\nu$ elements. This proves the following lemma.

\begin{lemma}\label{lemma:prime-orbit}
	The full gauged code dimension satisfies
	\begin{equation}\label{eq:E-prime-divisibility}
		0<K_\nu,\qquad p_\nu\mid K_\nu,\qquad p_\nu\ge2\nu+1\longrightarrow\infty.
	\end{equation}
	Namely, the gauged code dimension is positive and has a prime divisor tending to infinity along the family.
\end{lemma}

\subsection{Rate and distance of the gauged code}\label{appendix:gauged-parameters}

We now prove constant rate and linear distance for the gauged code. We use the distance and expansion bounds of $C$ to establish the needed estimates for $E$, then apply the gauged code distance result. All weights below count individual $\mbb F_q$ coordinates. The coordinates of the following spaces label physical $q$-dimensional sites and dressed $X$ stabilizer generators, respectively:
\begin{align}
	V & = C^2(X,\mathcal F)\oplus C^2(X,\mathcal G)\oplus E^3,\\
	U & = C^1(X,\mathcal F)\oplus C^1(X,\mathcal G)\oplus E^2.
\end{align}
Recall that $q=2^{s_q}$, so each $q$-dimensional site comprises $s_q$ binary sites, or qubits. Thus $n_\nu=\dim_{\mbb F_q}V$ is the number of $q$-dimensional sites, and $N_\nu=s_qn_\nu$ is the number of $\mbb F_2$ sites, equivalently the number of qubits. The covering has $v_\nu=|G_\nu|$ sheets, so every base cell has $v_\nu$ lifts. Since the local coefficient dimensions are fixed, $n_\nu=\Theta(v_\nu)$, and the number $m_\nu$ of local projector terms in $H_G$ is also $\Theta(v_\nu)$. Each term acts on $O(1)$ sites, and each site occurs in $O(1)$ terms.

\paragraph{Gauged code rate.}
In Eq.~\eqref{eq:gauge_condition_E}, fix $[b] = [b_0]$ and $[e] = 0$. The first two equations vanish, and the third becomes
\begin{align}
	I_\xi(a-a_0,b_0 - \rho b_0,\pi w) = 0 \qquad\text{for every cocycle }w\in E^2.
\end{align}
Recall that $\pi(y,x) = y$. This equation depends only on $[a-a_0]$ and the projected class $\pi_*[w]$: classes in $\ker\pi_*$ give the zero equation, and classes with the same projection give the same equation. By Eq.~\eqref{eq:fiber_invariant_image}, the projected classes range over exactly $H^2(X,\mathcal H)^\rho$. Since the expression is linear in this projected class, there are at most $\dim H^2(X,\mathcal H)^\rho$ independent linear equations on $H^2(X,\mathcal F)$. The class $[a_0]$ is a solution, so rank--nullity gives at least $q^{\dim H^2(X,\mathcal F)-\dim H^2(X,\mathcal H)^\rho}$ choices of $[a]$. Each choice gives a distinct compatible triple counted by $K_\nu$.

Since $p_\nu$ is odd, the average $p_\nu^{-1}\sum_{k=0}^{p_\nu-1} \rho^k$ commutes with the coboundary and identifies the cohomology of the invariant cochains with the invariant cohomology. Each lifted cell has an orbit of length $p_\nu$, so
\begin{align}
	\dim H^2(X,\mathcal H)^\rho
	\leq \dim C^2(X,\mathcal H)^\rho = \frac{\dim C^2(X,\mathcal H)}{p_\nu}.
\end{align}
Consequently,
\begin{equation}\label{eq:gauged_rate}
	\log_q K_\nu \geq \dim H^2(X,\mathcal F) - \frac{\dim C^2(X,\mathcal H)}{p_\nu} = \Omega(v_\nu).
\end{equation}
Here the first term is $\Omega(v_\nu)$ by the original lower bound on the rate given in Appendix~\ref{appendix:gauged_qLTC}, whereas the second is $O(v_\nu/p_\nu)$. Therefore, the final code has positive rate.

\paragraph{Distance and expansion of $E$.}
For a cochain complex $C^\bullet$ with coboundaries $\delta^j$ and dual boundaries $\partial_j=(\delta^{j-1})^T$, recall the cocycle and cycle expansions~\cite{Dinur2024sheaf,LWHL2609nonAbelian}:
\begin{align}\label{eq:cycle_expansions}
	\epsilon_\delta(j)=\min_{x\in C^j\setminus\ker\delta^j}
	\frac{|\delta^jx|}{\operatorname{dist}(x,\ker\delta^j)},\qquad
	\epsilon_\partial(j)=\min_{z\in C_j\setminus\ker\partial_j}
	\frac{|\partial_jz|}{\operatorname{dist}(z,\ker\partial_j)},
\end{align}
where $C_j=(C^j)^*$ and $\operatorname{dist}$ is Hamming distance in the chosen $\mbb F_q$ coordinates. 
When the expansion constants are positive, their definitions give the equivalent bounds 
\begin{align}\label{eq:syndrome_correction}
	\min_{\delta^jy=b}|y| \leq \epsilon_\delta(j)^{-1}|b|,
	\ b \in \im \delta^j, \qquad
	\min_{\partial_jy=b}|y| \leq \epsilon_\partial(j)^{-1}|b|,
	\ b \in \im \partial_j.
\end{align}
For example, if $b=\delta^jx$, the solutions of $\delta^jy=b$ are exactly $x + \ker\delta^j$, so their minimum weight is $\operatorname{dist}(x,\ker\delta^j)$.

Blocks $A$ and $B$ have linear code distances. With $\partial_{a,j}=(\delta_a^{j-1})^T$, their (co)cycle expansions $\epsilon_{\delta_a}(1),\epsilon_{\delta_a}(2),\epsilon_{\partial_a}(2)$, for $a=A,B$, are bounded below by a positive constant independent of $\nu$~\cite{gay2026asymptoticallygoodquantumlocally,BLN2026,CHLT2026}. For $E$, write $\delta_E^j$ for its coboundary from $E^j$ to $E^{j+1}$. Then we have the following lemma.

\begin{lemma}\label{lemma:E-parameters}
	There are constants $c_E,\epsilon_E>0$, independent of $\nu$, such that the distances satisfy
	\begin{align}\label{eq:E-distances}
		\min_{e \in \ker\delta_E^3 \setminus \im \delta_E^2}|e| \geq c_E v_\nu, \qquad
		\min_{\gamma \in \ker\partial_{E,j} \setminus \im \partial_{E,j+1}}|\gamma|\ge c_Ev_\nu \quad (j=2,3).
	\end{align}
	The cocycle expansions in degrees two and three and the cycle expansion in degree three satisfy
	\begin{align}\label{eq:E-expansion}
		\epsilon_{\delta_E}(2)\ge\epsilon_E,\qquad
		\epsilon_{\delta_E}(3)\ge\epsilon_E,\qquad
		\epsilon_{\partial_E}(3)\ge\epsilon_E.
	\end{align}
\end{lemma}
\begin{proof}
	Since the original third complex $C^\bullet(X,\mathcal H)$ has linear distances $\Omega(v_\nu)$ and its expansion constants $\epsilon_{\delta_C}(j)$ for $j=1,2,3$ and $\epsilon_{\partial_C}(j)$ for $j=2,3$ are bounded below by positive constant, we can apply these bounds successively to the two components in the definition of $\delta_E$.
	
	If $(y,x)\in E^3$ is a cocycle of sufficiently small weight, the original degree-three distance gives $y =\delta_C u$, with $|u| = O(|y|)$ by Eq.~\eqref{eq:syndrome_correction}. Subtracting $\delta_E(u,0)$ leaves $(0,x+(1-\rho)u)$. Its second component is a small degree-two cocycle and hence a coboundary, proving the first bound in Eq.~\eqref{eq:E-distances}.
	
	The boundary on the dual spaces is
	\begin{align}
		\partial_{E,j}(u,v) = \bigl(\partial_{C,j}u+(1- \rho)^Tv, \ \partial_{C,j-1}v\bigr).
	\end{align}
	For a small cycle $(u,v)$ in degree $j=2,3$, the original homological distance gives $v=\partial_{C,j}z$, with $|z|=O(|v|)$. Subtracting $\partial_{E,j+1}(0,z)$ leaves a small degree-$j$ cycle in the first summand and zero in the second. The remaining cycle is a boundary. This proves the second bound in Eq.~\eqref{eq:E-distances}.
	
	To bound $\epsilon_{\delta_E}(j)$ for $j=2,3$, take $(r,s)$ in the image of $\delta_E^j$ and solve $\delta_Cy=r$ with $|y|=O(|r|)$. The image condition gives $\delta_Cs=(1-\rho)r$, so $s+(1-\rho)y$ is a $j$-cocycle of weight $O(|r|+|s|)$. If this weight is below the original distance, solve $\delta_Cx=s+(1-\rho)y$ with $|x|=O(|r|+|s|)$. Otherwise $|r|+|s|=\Omega(v_\nu)$, and any preimage has weight $O(v_\nu)=O(|r|+|s|)$. Thus Eq.~\eqref{eq:syndrome_correction} gives a constant lower bound on $\epsilon_{\delta_E}(j)$.
	
	For $(r,s) \in \im \partial_{E,3}$, solve $\partial_{C,2}v = s$ with $|v|=O(|s|)$. The image condition gives $\partial_{C,2}r=(1 - \rho)^Ts$, so $r+(1- \rho)^Tv$ is a 2-cycle. Below the original homological distance, Eq.~\eqref{eq:syndrome_correction} gives $\partial_{C,3}u = r+(1- \rho)^Tv$ with $|u|=O(|r|+|s|)$. Otherwise any preimage has weight $O(v_\nu)=O(|r|+|s|)$. This proves the bound on $\epsilon_{\partial_E}(3)$, with constants independent of the order of $\rho$.
\end{proof}

Combine the three blocks $A,B,E$ by direct-sums $\delta_-=\delta_A^0\oplus\delta_B^0\oplus\delta_E^1$, $\delta=\delta_A^1\oplus\delta_B^1\oplus\delta_E^2$ and $\delta_+=\delta_A^2\oplus\delta_B^2\oplus\delta_E^3$, with domains and codomains
\begin{align}
	& \delta_-: C^0(X,\mathcal F)\oplus C^0(X,\mathcal G)\oplus E^1\longrightarrow U, \\
	& \delta: U \longrightarrow V, \\
	& \delta_+: V \longrightarrow C^3(X,\mathcal F)\oplus C^3(X,\mathcal G)\oplus E^4.
\end{align}
Since the three blocks have disjoint coordinates, their distance and expansion bounds give the same bounds for these direct sums, including
\begin{equation}\label{eq:check_distance}
	\min\{|\gamma|:\gamma \in \ker\delta_-^T \setminus \im \delta^T\} \geq c_* n_\nu
\end{equation}
for a constant $c_*>0$ independent of $\nu$.

\paragraph{Gauged code distance.}
It is proved in Ref.~\cite{LWHL2609nonAbelian} using the general Knill--Laflamme condition that, when the ingredient code blocks have good distances and cocycle expansions, the gauged code inherits good distance. Since $A,B$ and $E$ satisfy these conditions, our gauged code family has distance $\Omega(N_\nu)$. The constant soundness for the gauged code is prove in Section \ref{subsection:gauged-operator-bound} in the main text.

\section{Binary NLTM from ternary stabilizer codes}\label{appendix:ternary}

A qutrit stabilizer code encoding $k\ge1$ logical qutrits supplies a logical dimension $3^k$. A linear energy barrier protects this factor at low energy and yields binary NLTM after embedding each qutrit into two qubits.
Appendix~\ref{appendix:hamiltonian} defines the embedded Hamiltonian and states the general energy barrier theorem.
Appendix~\ref{appendix:input} realizes its requirements with ternary Tanner codes, and Appendix~\ref{appendix:band} gives the final proof.
This construction separates binary and intrinsic NLTM, as explained in Appendix~\ref{appendix:ternary-properties}.

\subsection{Energy barrier criterion for binary NLTM}\label{appendix:hamiltonian}

Let $S_1,\cdots,S_m$ be a chosen list of commuting ternary Pauli checks on $n$ qutrits. The check Hamiltonian is
\begin{equation}\label{equation:bar-check-projector}
	h_i=I-\frac{I+S_i+S_i^2}{3},\quad H_n^{\mathrm{qutrit}}=\sum_{i=1}^{m}h_i.
\end{equation}

\begin{definition}[Energy barrier]\label{definition:energy-barrier}
	For a chosen list $S_1,\cdots,S_m$ of commuting ternary Pauli checks on $n$ qutrits, set $\omega=e^{2\pi i/3}$. For a ternary Pauli $P$, define its syndrome relative to this list, $\sigma(P)\in\mbb{F}_3^m$, by $S_iP=\omega^{\sigma(P)_i}PS_i$. A \emph{single-site Pauli path} is a sequence $P_0=I,P_1,\cdots,P_T$ such that $P_jP_{j-1}^\dagger$ is supported on at most one ternary site. When a nontrivial logical Pauli exists, the energy barrier of the corresponding check Hamiltonian is
	\begin{equation}
		B\coloneqq\min_{\substack{\{P_j\}_{j=0}^T\text{ single-site Pauli path}\\ P_T\text{ nontrivial logical Pauli}}}\max_{0\le j\le T}|\sigma(P_j)|.
	\end{equation}
\end{definition}

Recall the qutrit Hamiltonian $H_n^{\mathrm{qutrit}}=\sum_i h_i$ in Eq.~\eqref{equation:bar-check-projector}. Embed each qutrit into two qubits by $J\ket{0}=\ket{00}$, $J\ket{1}=\ket{01}$, and $J\ket{2}=\ket{10}$, and let $\Pi_v\coloneqq JJ^\dagger=I-\ketbra{11}{11}_v$ project onto the legal subspace at site $v$. For a set $R$ of sites, write $J_R\coloneqq\bigotimes_{v\in R}J$ and let $I_d$ be the identity on $\mbb{C}^d$. For each check, set $R_i\coloneqq\operatorname{supp}(h_i)\subset[n]$ and write $h_i=h_i^{(R_i)}\otimes\bigotimes_{v\notin R_i}I_3$. Its qubit embedding and the Hamiltonian are
\begin{align}
	\widehat h_i&\coloneqq J_{R_i}h_i^{(R_i)}J_{R_i}^\dagger\otimes\bigotimes_{v\notin R_i}I_4,\\
	H_n&\coloneqq\sum_i\widehat h_i+\sum_{v=1}^n\ketbra{11}{11}_v.\label{equation:main-hamiltonian}
\end{align}
Thus $\widehat h_i$ vanishes if any site in $R_i$ is illegal, while illegal sites outside $R_i$ do not affect it. The physical Hilbert space decomposes according to the set $F$ of illegal sites:
\begin{equation}\label{equation:ham-sector-space}
	(\mbb{C}^2)^{\otimes2n}=\bigoplus_{F\subset[n]}\mc{H}_F,\quad \mc{H}_F\coloneqq\bigotimes_{v\in F}\operatorname{span}\{\ket{11}_v\}\otimes\bigotimes_{v\notin F}\im J_v.
\end{equation}

\begin{proposition}[Properties of the embedded Hamiltonian]
	Let the chosen nonidentity Pauli checks generate the stabilizer group of a ternary $[\![n,k]\!]$ code. Suppose each check is supported on at most $w$ sites and each site is incident on at most $\Delta$ checks. Set $\ell\coloneqq2\max\{w,1\}$. The Hamiltonian $H_n$ in Eq.~\eqref{equation:main-hamiltonian} has the following properties:
	\begin{enumerate}[label=(\roman*)]
		\item All terms are mutually commuting orthogonal projectors of norm one, each supported on at most $\ell$ qubits, and each qubit belongs to at most $\Delta+1$ terms.
		\item The Hamiltonian satisfies
		\begin{equation}
			\operatorname{spec}(H_n)\subset\mbb{Z}_{\ge0},\quad 0\le H_n\le(\Delta+1)nI,
		\end{equation}
		where $\operatorname{spec}(H_n)$ denotes its spectrum.
		\item The Hamiltonian is frustration free, with
		\begin{equation}
			\ker H_n=J_{[n]}\bigl(\ker H_n^{\mathrm{qutrit}}\bigr),\quad \dim\ker H_n=3^k.
		\end{equation}
		\item The subspaces $\mc{H}_F$ form an orthogonal decomposition, are invariant under every term, and satisfy
		\begin{equation}\label{equation:ham-sector}
			H_n\big|_{\mc{H}_F}=|F|I_{\mc{H}_F}+\sum_{i:\,R_i\cap F=\varnothing}\widehat h_i\big|_{\mc{H}_F}.
		\end{equation}
	\end{enumerate}
\end{proposition}

\comments{
	\begin{proof}
		Since $J^\dagger J=I_3$ and $J^\dagger\ket{11}=0$, isometric conjugation preserves the projection property and norm of each nonzero $h_i$, and $\widehat h_i$ commutes with every $\Pi_v$. The penalties $I-\Pi_v$ are also norm-one projectors and commute with every term. Thus every term preserves the sectors in Eq.~\eqref{equation:ham-sector-space}. On $\mc{H}_F$, the penalties sum to $|F|I$, and precisely the checks meeting $F$ vanish. The surviving checks are the original commuting qutrit checks on the legal sites, conjugated by $J_{[n]\setminus F}$, which proves both their commutation and Eq.~\eqref{equation:ham-sector}. The support and incidence bounds give locality $\ell$ and degree at most $\Delta+1$.
		
		Simultaneous diagonalization gives the nonnegative integer spectrum. There are at most $n\Delta$ checks because each uses at least one site, so the $n$ penalties give between $n$ and $(\Delta+1)n$ terms and the stated operator bound. Finally, positivity forces every ground state to satisfy every penalty and therefore to lie in the fully legal sector. On this sector, $H_nJ_{[n]}=J_{[n]}H_n^{\mathrm{qutrit}}$, so its kernel is exactly the embedded ternary code space. This space is nonzero and has dimension $3^k$, proving frustration freedom and the ground space claim.
	\end{proof}
}

\begin{proof}
	Each $\widehat h_i$ is a norm-one projector commuting with every $\Pi_v$. On $\mc H_F$, the penalties contribute $|F|I$, checks meeting $F$ vanish, and the remaining checks are the original commuting qutrit checks conjugated by $J_{[n]\setminus F}$. This proves commutation and Eq.~\eqref{equation:ham-sector}.
	
	The support and incidence bounds give locality $\ell$ and degree at most $\Delta+1$. Simultaneous diagonalization gives the nonnegative integer spectrum, and the at most $n\Delta$ checks together with the $n$ penalties give the operator bound. Positivity forces every ground state into the fully legal sector, where $H_nJ_{[n]}=J_{[n]}H_n^{\mathrm{qutrit}}$; hence the ground space is the embedded ternary code and has dimension $3^k$.
\end{proof}

The energy barrier $B_n$ controls logical operations assembled from local errors, including paths whose total support exceeds the code distance.

A family is qLDPC with respect to the chosen check lists when both check weight and the number of checks incident on a qutrit are uniformly bounded.

\begin{theorem}[Magic complexity from a linear energy barrier]\label{theorem:main}
	Consider a family of $[\![n,k]\!]$ qutrit stabilizer codes with unbounded block length $n$ and $k\ge1$ logical qutrits. For each code, fix a (possibly redundant) list of nonidentity ternary Pauli checks generating its stabilizer group, and let $H_n$ be the corresponding qubit Hamiltonian in Eq.~\eqref{equation:main-hamiltonian}. Assume the family is qLDPC with respect to these lists and has energy barrier $B_n\ge\beta_0n$ for some $\beta_0>0$ and all sufficiently large $n$. Then there are constants $\eta,c>0$ such that for all sufficiently large $n$ and every $2n$-qubit state $\rho$,
	\begin{equation}
		\frac{\Tr(H_n\rho)}{2n}\le\eta \quad\Longrightarrow\quad \Cmagicd{2}(\rho)>\frac15\log  n+\log  c.
	\end{equation}
\end{theorem}

The conclusion allows arbitrary finite stabilizer ancillas, discarded outputs, and classical randomization, and extends to gates acting on any fixed number of qubits with a constant change in the depth bound.

\subsection{Tanner codes with a linear energy barrier}\label{appendix:input}

We first derive a linear energy barrier from small-set expansion, then specify the Tanner family. Consider the check matrices $H_X\in\mbb{F}_3^{m_X\times n}$ and $H_Z\in\mbb{F}_3^{m_Z\times n}$ with
\begin{equation}\label{equation:input-css}
	H_XH_Z^T=0,\quad \mathsf S_X=\operatorname{row}H_X,\quad \mathsf S_Z=\operatorname{row}H_Z,\quad k=n-\dim\mathsf S_X-\dim\mathsf S_Z\ge1.
\end{equation}
Let $X$ and $Z$ be the qutrit Pauli operators, and set $\omega=e^{2\pi i/3}$. For $a,b\in\mbb{F}_3^n$, define
\begin{equation}
	X(a)=\bigotimes_v X^{a_v},\quad Z(b)=\bigotimes_v Z^{b_v},
\end{equation}
with global phases omitted. The rows of $H_X$ and $H_Z$ define $m=m_X+m_Z$ stabilizer generators $S_1,\cdots,S_m$ and generate
\begin{equation}
	\mathsf S=\bigl\langle X(x),Z(z):x\in\mathsf S_X,\ z\in\mathsf S_Z\bigr\rangle.
\end{equation}
\comments{
	Eq.~\eqref{equation:input-css} makes these checks commute, and $h_i$ in Eq.~\eqref{equation:bar-check-projector} projects onto the eigenspaces on which the $i$th check is violated. The syndrome $\sigma(P)\in\mbb{F}_3^m$ of a Pauli $P$ is specified by
	\begin{equation}
		S_iP=\omega^{\sigma(P)_i}PS_i.
	\end{equation}
	Thus $P$ applied to a code state has energy $|\sigma(P)|$. For $P=X(a)Z(b)$,
	\begin{equation}\label{equation:bar-pauli-energy}
		|\sigma(P)|=|H_Za|+|H_Xb|.
	\end{equation}
}

Eq.~\eqref{equation:input-css} makes these checks commute. For $P=X(a)Z(b)$, Definition~\ref{definition:energy-barrier} gives the energy of $P$ applied to a code state as
\begin{equation}\label{equation:bar-pauli-energy}
	|\sigma(P)|=|H_Za|+|H_Xb|.
\end{equation}
For $y\in\mbb{F}_3^n$ and $\mathsf A\subset\mbb{F}_3^n$, define
\begin{equation}
	\operatorname{dist}(y,\mathsf A)=\min_{z\in\mathsf A}|y-z|,\quad r_X(a)=\operatorname{dist}(a,\mathsf S_X),\quad r_Z(b)=\operatorname{dist}(b,\mathsf S_Z).
\end{equation}

\begin{definition}[Small-set expansion]
	For $0<\alpha\le1$ and $\beta>0$, the check matrices $H_X$ and $H_Z$ have $(\alpha,\beta)$-small-set expansion if every $y\in\mbb{F}_3^n$ with $|y|\le\alpha n$ satisfies
	\begin{subequations}\label{equation:input-expansion}
		\begin{align}
			\frac{|H_Zy|}{m_Z}&\ge\beta\frac{\operatorname{dist}(y,\mathsf S_X)}{n},\\
			\frac{|H_Xy|}{m_X}&\ge\beta\frac{\operatorname{dist}(y,\mathsf S_Z)}{n}.
		\end{align}
	\end{subequations}
\end{definition}

If $r_X(a)\le\alpha n$, choose $x\in\mathsf S_X$ such that $|a-x|=r_X(a)$. Since $H_Zx=0$, applying small-set expansion to $a-x$ gives
\begin{subequations}\label{equation:bar-reduced-expansion}
	\begin{align}
		r_X(a)\le\alpha n&\ \Longrightarrow\ |H_Za|\ge\beta\frac{m_Z}{n}r_X(a),\\
		r_Z(b)\le\alpha n&\ \Longrightarrow\ |H_Xb|\ge\beta\frac{m_X}{n}r_Z(b).
	\end{align}
\end{subequations}

\begin{proposition}[Linear energy barrier from small-set expansion]
	\label{proposition:bar-linear}
	Suppose the family has $(\alpha,\beta)$-small-set expansion and $\min\{m_X,m_Z\}\ge c_m n$ for constants $\alpha,\beta,c_m>0$ independent of $n$. Then, for every sufficiently large code in the family,
	\begin{equation}\label{equation:bar-lower-bound}
		B_n\ge\beta\frac{\min\{m_X,m_Z\}}{n}\left\lfloor\frac{\alpha n}{2}\right\rfloor\ge\frac{\alpha\beta c_m}{4}n.
	\end{equation}
\end{proposition}

\begin{proof}
	Consider a single-site Pauli path $P_0=I,P_1,\cdots,P_T$ ending at a nontrivial logical Pauli, and write $P_j=X(a_j)Z(b_j)$ up to phase. If $a_T\notin\mathsf S_X$, then $H_Za_T=0$ and Eq.~\eqref{equation:bar-reduced-expansion} imply $r_X(a_T)>\alpha n$. Since $r_X(a_0)=0$ and each step changes $r_X(a_j)$ by at most one, the path contains an index $j$ with $r_X(a_j)=L\coloneqq\lfloor\alpha n/2\rfloor$. At that index, Eqs.~\eqref{equation:bar-pauli-energy} and \eqref{equation:bar-reduced-expansion} give $|\sigma(P_j)|\ge\beta(m_Z/n)L$. Otherwise $b_T\notin\mathsf S_Z$, and the same argument gives $|\sigma(P_j)|\ge\beta(m_X/n)L$ at some index. Thus every path reaches syndrome weight at least $\beta\min\{m_X,m_Z\}L/n$, and $L\ge\alpha n/4$ for $n\ge4/\alpha$ proves Eq.~\eqref{equation:bar-lower-bound}.
\end{proof}

The analogous energy barrier for binary Tanner codes was proved in~\cite{williamson2024layercodes}.

\comments{
	The quantum Tanner construction over finite fields in~\cite[Sec.~3]{golowich2023nltshamiltoniansstronglyexplicitsos} provides the ternary family required by Proposition~\ref{proposition:bar-linear}.
}

The formation of quantum Tanner code~\cite{leverrier2022quantumtannercodes} in the general finite fields~\cite{golowich2023nltshamiltoniansstronglyexplicitsos} provides the ternary family required by Proposition~\ref{proposition:bar-linear}.

\comments{
	For the field $\mbb{F}_3$ and inner code rates in $(1/8,7/8)$, Corollary~11 of~\cite{golowich2023nltshamiltoniansstronglyexplicitsos} gives a product expansion constant $\rho_{\mathrm{in}}>0$ independent of the inner length. Choose a sufficiently large power of two $Q=2^h$, with $h\ge4$, and set $r=Q+1$ and $a=\lfloor r/4\rfloor$, so that $r$ meets the degree threshold in Theorem~16 of that reference for $\rho_{\mathrm{in}}$ and Corollary~11 ensures simultaneous product expansion for dimensions $a$ and $r-a$. Fix codes $C_A,C_B\subset\mbb{F}_3^r$ with $\dim C_A=a$ and $\dim C_B=r-a$ such that both $(C_A,C_B)$ and $(C_A^\perp,C_B^\perp)$ are $\rho_{\mathrm{in}}$-product expanding.
}

Fix a sufficiently large power of two $Q=2^h$, and set $r=Q+1$ and $a=\lfloor r/4\rfloor$. Ref.~\cite{golowich2023nltshamiltoniansstronglyexplicitsos} gives codes $C_A,C_B\subset\mbb F_3^r$ of dimensions $a,r-a$ such that both $(C_A,C_B)$ and $(C_A^\perp,C_B^\perp)$ are $\rho_{\mathrm{in}}$-product expanding, for a constant $\rho_{\mathrm{in}}>0$ independent of $r$. Choose $r$ large enough for the small-set expansion bound in that reference, and keep $Q,C_A,C_B$ fixed throughout the family.

\comments{
	For each $j\ge1$, take an inverse-closed generating set $S_j=S_j^{-1}$ for the $r$-regular Ramanujan Cayley graph on $G_j$ supplied by Theorem~12 of~\cite{golowich2023nltshamiltoniansstronglyexplicitsos}, and set $A_j=B_j=S_j$ in the Tanner construction. The resulting code has
}

For each $j\ge1$, take an inverse-closed generating set $S_j=S_j^{-1}$ for the $r$-regular Ramanujan Cayley graph on $G_j$ supplied by Ref.~\cite{golowich2023nltshamiltoniansstronglyexplicitsos}, and set $A_j=B_j=S_j$ in the Tanner construction. The resulting code has
\begin{equation}\label{equation:input-sizes}
	|G_j|=Q^{2j}(Q^{4j}-1),\quad n_j=|G_j|r^2=r^2Q^{2j}(Q^{4j}-1).
\end{equation}
The local check spaces for $X$ and $Z$ are $C_A\otimes C_B$ and $C_A^\perp\otimes C_B^\perp$, respectively, both of dimension $a(r-a)$. Choosing a basis in each local check space gives
\begin{equation}\label{equation:input-counts}
	m_X=m_Z=2|G_j|a(r-a)=c_m n_j,\quad c_m=\frac{2a(r-a)}{r^2}>0.
\end{equation}
The logical dimension satisfies
\begin{equation}
	k_j\ge n_j-m_X-m_Z=\left(1-\frac{2a}{r}\right)^2n_j\ge\frac{n_j}{4}.
\end{equation}
Moreover, every check has weight at most $r^2$, and every site belongs to at most $4a(r-a)$ checks. Thus $w$ and $\Delta$ are independent of $j$. Since each local check space has constant dimension, the choice of basis changes syndrome weight by at most a constant factor.

\comments{
	Theorem~16 of~\cite{golowich2023nltshamiltoniansstronglyexplicitsos} implies that these codes have $(\alpha,\beta)$-small-set expansion for constants $\alpha,\beta>0$ independent of $j$. In the notation of that reference, $C_X=\ker H_X$ and $C_Z=\ker H_Z$, so $C_X^\perp=\mathsf S_X$ and $C_Z^\perp=\mathsf S_Z$, identifying its expansion bounds with the pair of inequalities in~\eqref{equation:input-expansion}. Together with~\eqref{equation:input-counts}, Proposition~\ref{proposition:bar-linear} gives $B_{n_j}=\Omega(n_j)$, verifying the hypotheses of Theorem~\ref{theorem:main} for this Tanner family.
}

Ref.~\cite{golowich2023nltshamiltoniansstronglyexplicitsos} implies that these codes have $(\alpha,\beta)$-small-set expansion for constants $\alpha,\beta>0$ independent of $j$. In the notation of that reference, $C_X=\ker H_X$ and $C_Z=\ker H_Z$, so $C_X^\perp=\mathsf S_X$ and $C_Z^\perp=\mathsf S_Z$, identifying its expansion bounds with the pair of inequalities in~\eqref{equation:input-expansion}. Together with~\eqref{equation:input-counts}, Proposition~\ref{proposition:bar-linear} gives $B_{n_j}=\Omega(n_j)$, verifying the hypotheses of Theorem~\ref{theorem:main} for this Tanner family.

\begin{corollary}[NLTM from Tanner codes]\label{corollary:tanner-obstruction}
	The ternary quantum Tanner codes~\cite{golowich2023nltshamiltoniansstronglyexplicitsos}, specified above, satisfy the hypotheses of Theorem~\ref{theorem:main}. Consequently, there exist a fixed constant $\varepsilon>0$ and a family of Hamiltonians $\{H_N\}$ of bounded degree and locality, defined for all sufficiently large $N$, such that each $H_N$ acts on $N$ qubits, consists of $\Theta(N)$ commuting projectors of norm one, and has zero ground energy. Every $N$-qubit state $\rho_N$ satisfying $\Tr(H_N\rho_N)\le\varepsilon N$ has binary magic circuit complexity $\Omega(\ln N)$.
\end{corollary}

\subsection{Low-energy logical protection and NLTM}\label{appendix:band}

A linear energy barrier protects a logical factor throughout the full low-energy sector, including syndrome components disconnected from the ground syndrome and sectors with illegal qubit pairs. The proof uses the chosen ternary stabilizer check list, its incidence bound $\Delta$, and the qubit Hamiltonian with commuting projector terms. A \emph{site} means one qutrit or its corresponding qubit pair, so an operator on at most $u$ qubits touches at most $u$ sites. All statements concern compression, without requiring local operators to preserve the compressed subspace.

Error-corrected and dressed logical observables characterize protection in quantum memories~\cite{chesi2009thermodynamicstabilitycriteriaquantum,bombin2012selfcorrectingquantumcomputers}, and quantum encodings in excited-state regions of qLDPC Hamiltonians have been studied under linear confinement~\cite{placke2024topologicalquantumspinglass}. Energy barriers also enter constructions of decoders and partial self-correction~\cite{gu2026layercodespartiallyselfcorrecting}.

\begin{theorem}[Logical protection in the low-energy sector]\label{theorem:band-full}
	Fix an $[\![n,k]\!]$ qutrit stabilizer code with $k\ge1$ and a possibly redundant Pauli check list generating its stabilizer group, with each qutrit incident on at most $\Delta\in\mbb{Z}_{>0}$ checks. Let $B_n$ denote the energy barrier of this check list. Let $H_n$ be the corresponding qubit Hamiltonian in Eq.~\eqref{equation:main-hamiltonian}. Let $e,u$ be nonnegative integers satisfying
	\begin{equation}
		B_n>\bigl(3e+\lfloor u/2\rfloor\bigr)\Delta,
	\end{equation}
	and set $Q_e\coloneqq\mathbf1_{[0,e]}(H_n)$. Then there exist a nonzero finite-dimensional Hilbert space $\mc{G}_e$ and an isometry $W_e:\mbb{C}^{3^k}\otimes\mc{G}_e\longrightarrow(\mbb{C}^2)^{\otimes 2n}$ with image $\im Q_e$, such that for every operator $O$ supported on at most $u$ qubits,
	\begin{equation}\label{equation:band-full-compression}
		W_e^\dagger O W_e=I_{3^k}\otimes A_{e,O}
	\end{equation}
	for some operator $A_{e,O}$ on $\mc{G}_e$.
\end{theorem}

We first align the logical bases in all qutrit syndrome sectors. We then repair each sector indexed by a set of illegal sites into an enlarged legal sector and identify its image as the same logical factor tensored with an auxiliary subspace. The final step controls operators between sectors indexed by different sets of illegal sites.

Let $S_1,\cdots,S_m$ be the chosen check list, let $\mathsf S$ be its stabilizer group, and set $\omega=e^{2\pi i/3}$. For $s\in\mbb{F}_3^m$, define
\begin{equation}
	\mc{C}_s\coloneqq\{\ket{\psi}\in(\mbb{C}^3)^{\otimes n}:S_i\ket{\psi}=\omega^{s_i}\ket{\psi}\text{ for all }i=1,\cdots,m\}.
\end{equation}
Write $\Xi\coloneqq\{s:\mc{C}_s\ne\{0\}\}$ for the realizable syndromes and, for $t\ge0$, set
\begin{equation}
	\Xi_t\coloneqq\{s\in\Xi:|s|\le t\},\quad \mc{Q}_t^{\mathrm{qutrit}}\coloneqq\bigoplus_{s\in\Xi_t}\mc{C}_s,\quad \mc{G}_t^{\mathrm{qutrit}}\coloneqq\operatorname{span}\{\ket{s}:s\in\Xi_t\}.
\end{equation}
The syndrome weight counts every violated check in the chosen list, including redundant checks, so $\mc{Q}_t^{\mathrm{qutrit}}$ is the entire qutrit low-energy sector up to energy $t$. The vectors $\ket{s}$ form an abstract orthonormal basis.

\begin{lemma}[Subsystem code from energy barrier]\label{lemma:band-legal-factor}
	Let $t\ge0$ and let $q$ be a nonnegative integer satisfying
	\begin{equation}\label{equation:band-legal-barrier}
		B_n>2t+\lfloor q/2\rfloor\Delta.
	\end{equation}
	Then there is an isometry $W_{\mathrm{qutrit}}:\mbb{C}^{3^k}\otimes\mc{G}_t^{\mathrm{qutrit}}\longrightarrow\mc{Q}_t^{\mathrm{qutrit}}$ onto the full qutrit low-energy sector, with $W_{\mathrm{qutrit}}(\mbb{C}^{3^k}\otimes\ket{s})=\mc{C}_s$ for every $s\in\Xi_t$, such that every operator $T$ supported on at most $q$ qutrits satisfies
	\begin{equation}
		W_{\mathrm{qutrit}}^\dagger T W_{\mathrm{qutrit}}=I_{3^k}\otimes A_T
	\end{equation}
	for some operator $A_T$ on $\mc{G}_t^{\mathrm{qutrit}}$.
\end{lemma}

\begin{proof}
	\comments{
		After relabeling the checks, let $S_1,\cdots,S_{n-k}$ be independent generators. For $y\in\mbb{F}_3^{n-k}$, their joint eigenspace with eigenvalues $\omega^{y_j}$ has projector
		\begin{equation}
			\Pi_y=\frac{1}{3^{n-k}}\sum_{\alpha\in\mbb{F}_3^{n-k}}\omega^{-\alpha\cdot y}\prod_{j=1}^{n-k}S_j^{\alpha_j}.
		\end{equation}
		Only the identity term has nonzero trace, so $\Tr\Pi_y=3^k$. Writing $S_i=\prod_{j=1}^{n-k}S_j^{c_{ij}}$ with $c_{ij}\in\mbb{F}_3$, the full realizable syndrome space is
		\begin{equation}
			\Xi=\left\{\left(\sum_{j=1}^{n-k}c_{ij}y_j\right)_{i=1}^m:y\in\mbb{F}_3^{n-k}\right\}.
		\end{equation}
		The first $n-k$ entries recover $y$, so this parametrization is injective and $\Xi$ is an $(n-k)$-dimensional linear subspace. Every realizable sector is the image of exactly one $\Pi_y$ and therefore has dimension $3^k$. Each such syndrome is attained by a Pauli: for nonzero $\ket{\phi}\in\mc{C}_0$ and $\ket{\psi}\in\mc{C}_s$, a Pauli expansion of $\ketbra{\psi}{\phi}$ supplies $P$ with $\bra{\psi}P\ket{\phi}\ne0$, which forces $\sigma(P)=s$ by syndrome orthogonality.
	}

	Each realizable syndrome $s$ is attained by a Pauli: for nonzero $\ket{\phi}\in\mc{C}_0$ and $\ket{\psi}\in\mc{C}_s$, a Pauli expansion of $\ketbra{\psi}{\phi}$ supplies $P$ with $\bra{\psi}P\ket{\phi}\ne0$, which forces $\sigma(P)=s$ by syndrome orthogonality. Then $P\mc{C}_0=\mc{C}_s$, so every realizable syndrome sector has dimension $3^k$.

	Form a directed multigraph on $\Xi_t$, with an edge $s\xrightarrow{P}s+\sigma(P)$ for every Pauli $P$ supported on at most $q$ sites whenever both endpoints lie in $\Xi_t$. The edge records $P\mc{C}_s=\mc{C}_{s+\sigma(P)}$ and has reverse edge labelled by $P^\dagger$. The graph includes loops and parallel edges and need not be connected.
	
	\begin{claim}\label{claim:band-holonomy}
		Under Eq.~\eqref{equation:band-legal-barrier}, the product of the labels of every closed path in this graph is a stabilizer up to a scalar phase.
	\end{claim}
	
	We prove the claim below. Identify $\mc{C}_0$ with $\mc{L}=\mbb{C}^{3^k}$ by fixing an orthonormal basis. For every connected component $C$, choose a reference syndrome $s_C$ and a Pauli $E_C$ with syndrome $s_C$. No support or energy bound is imposed on $E_C$, which only identifies $\mc{C}_0$ with $\mc{C}_{s_C}$. For $s\in C$, choose a path from $s_C$ to $s$ with Pauli product $U_s$ and define
	\begin{equation}
		W_{\mathrm{qutrit}}(\ket{\xi}\otimes\ket{s})\coloneqq U_sE_C\ket{\xi},\quad \ket{\xi}\in\mc{L}.
	\end{equation}
	Each $U_sE_C$ maps $\mc{C}_0$ unitarily onto $\mc{C}_s$, and the syndrome sectors are orthogonal, so $W_{\mathrm{qutrit}}$ is an onto isometry.
	\comments{
		The choice of paths affects only the common phase of each transported logical basis. Indeed, if $\widetilde U_s$ is the product along another path from $s_C$ to $s$, then $\widetilde U_s^\dagger U_s$ is the product along a closed path. Claim~\ref{claim:band-holonomy} makes it a stabilizer up to phase, which acts as a scalar on $\mc{C}_{s_C}$. Consequently,
		\begin{equation}
			U_sE_C\ket{\xi}=c\,\widetilde U_sE_C\ket{\xi},\quad \ket{\xi}\in\mc{C}_0,
		\end{equation}
		where $|c|=1$ is independent of the logical vector.
	}

	For a Pauli $P$ supported on at most $q$ sites, its block from $\mc{C}_s$ to $\mc{C}_{s'}$ vanishes unless $s'=s+\sigma(P)$. In the latter case the corresponding edge lies in one component $C$, and $U_{s'}^\dagger P U_s$ is the product along a closed path based at $s_C$. Write it as $\omega^aS$ with $S\in\mathsf S$ using the claim. Since $S$ acts as a scalar on $\mc{C}_{s_C}$, the logical block is
	\begin{equation}
		E_C^\dagger U_{s'}^\dagger P U_sE_C\big|_{\mc{C}_0}=\omega^aE_C^\dagger S E_C\big|_{\mc{C}_0}=c_{s',s,P}I_{\mc{L}}.
	\end{equation}
	\comments{
		Set $c_{s',s,P}=0$ for the vanishing blocks and define $A_P\coloneqq\sum_{s,s'\in\Xi_t}c_{s',s,P}\ketbra{s'}{s}$. All blocks between different graph components vanish, and the preceding calculation treats every nonzero block within a component. Summing over the full syndrome basis gives
		\begin{equation}
			W_{\mathrm{qutrit}}^\dagger P W_{\mathrm{qutrit}}=\sum_{s,s'\in\Xi_t}(c_{s',s,P}I_{\mc{L}})\otimes\ketbra{s'}{s}=I_{\mc{L}}\otimes A_P.
		\end{equation}
	}

	All blocks between different graph components vanish, and every remaining block is scalar on $\mc L$, so $W_{\mathrm{qutrit}}^\dagger P W_{\mathrm{qutrit}}=I_{\mc L}\otimes A_P$ for some $A_P$.
	Expanding any $q$-local operator $T=\sum_P\alpha_PP$ in the Pauli basis preserves the support bound and gives the required identity with $A_T=\sum_P\alpha_PA_P$.
\end{proof}

\begin{proof}[Proof of Claim~\ref{claim:band-holonomy}]
	Let the successive vertices and labels of a closed path be $s_0,s_1,\cdots,s_L=s_0$ and $P_1,\cdots,P_L$. Set $V_0=I$ and $V_j=P_j\cdots P_1$. The successive syndrome differences telescope:
	\begin{equation}
		\sigma(V_j)=\sum_{r=1}^j\sigma(P_r)=\sum_{r=1}^j(s_r-s_{r-1})=s_j-s_0.
	\end{equation}
	Thus the same accumulated Pauli, applied to a ground state, has syndrome weight
	\begin{equation}\label{equation:band-translated-path}
		|\sigma(V_j)|\le|s_j|+|s_0|\le2t.
	\end{equation}
	In particular, $V_L$ has zero syndrome and preserves the code space.
	
	To obtain a single-site Pauli path, write an edge label of support $1\le h\le q$ as $P_j=P_{j,h}\cdots P_{j,1}$, with one factor per site and its scalar phase absorbed into the first factor. After $a$ factors, let $V_{j,a}\coloneqq P_{j,a}\cdots P_{j,1}V_{j-1}$. Each factor changes at most $\Delta$ syndrome entries. Comparing with both endpoints of the edge and using Eq.~\eqref{equation:band-translated-path} yields
	\begin{equation}\label{equation:band-edge-bound}
		|\sigma(V_{j,a})|\le\min\{2t+a\Delta,2t+(h-a)\Delta\}\le2t+\lfloor q/2\rfloor\Delta.
	\end{equation}
	An edge proportional to $I$ is a step of zero support. The resulting path starts at $I$, ends at $V_L$, and has maximum syndrome weight bounded by Eq.~\eqref{equation:band-edge-bound}. If $V_L$ had nontrivial logical action, this would contradict the definition of $B_n$ and Eq.~\eqref{equation:band-legal-barrier}. Thus $V_L$ is a stabilizer up to phase.
\end{proof}

We now extend this protection to the full qubit low-energy sector. Let $\mc{Q}_t^{\mathrm{legal}}\coloneqq J_{[n]}\mc{Q}_t^{\mathrm{qutrit}}$ be the embedded legal low-energy sector.

\begin{proof}[Proof of Theorem~\ref{theorem:band-full}]
	The orthogonal sectors $\mc{H}_F$ in Eq.~\eqref{equation:ham-sector-space}, indexed by the sets $F$ of illegal sites, are invariant under $H_n$, so
	\begin{equation}
		\im Q_e=\bigoplus_{F\subset[n]}\bigl(\mc{H}_F\cap\im Q_e\bigr).
	\end{equation}
	Separate the embedded checks by whether their supports $R_i$ meet $F$:
	\begin{equation}
		H_n^{\overline F}\coloneqq\sum_{i:R_i\cap F=\varnothing}\widehat h_i,\quad H_n^{\cap F}\coloneqq\sum_{i:R_i\cap F\ne\varnothing}\widehat h_i.
	\end{equation}
	By Eq.~\eqref{equation:ham-sector}, $H_n=|F|I+H_n^{\overline F}$ on $\mc{H}_F$, so sectors with $|F|>e$ contribute nothing to $\im Q_e$. For $|F|\le e$, let $R_F$ swap $\ket{11}$ and $\ket{00}$ on $F$, fix the other two basis states, and act trivially elsewhere. Then $R_F^\dagger=R_F$, $R_F^2=I$, and $[R_F,H_n^{\overline F}]=0$.
	
	\begin{claim}\label{claim:band-repair}
		For every $F$ with $|F|\le e$,
		\begin{equation}\label{equation:band-repair-inclusion}
			R_F\bigl(\mc{H}_F\cap\im Q_e\bigr)\subset\mc{Q}_{e\Delta}^{\mathrm{legal}}.
		\end{equation}
	\end{claim}
	
	We prove this claim below. Apply Lemma~\ref{lemma:band-legal-factor} with $t=e\Delta$ and $q=2e+u$, for which its hypothesis is exactly
	\begin{equation}
		B_n>2e\Delta+\lfloor(2e+u)/2\rfloor\Delta=\bigl(3e+\lfloor u/2\rfloor\bigr)\Delta.
	\end{equation}
	Set $\mc{L}=\mbb{C}^{3^k}$ and
	\begin{equation}
		W\coloneqq J_{[n]}W_{\mathrm{qutrit}}:\mc{L}\otimes\mc{G}_{e\Delta}^{\mathrm{qutrit}}\longrightarrow\mc{Q}_{e\Delta}^{\mathrm{legal}}.
	\end{equation}
	This is an onto isometry with $W(\mc{L}\otimes\ket{s})=J_{[n]}\mc{C}_s$. If a qubit operator $\widetilde O$ touches at most $q$ sites, its pullback $J_{[n]}^\dagger\widetilde OJ_{[n]}$ is supported on those same qutrits, because $J_{[n]}$ is a tensor product of single-site isometries. Consequently,
	\begin{equation}\label{equation:band-legal-compression}
		W^\dagger\widetilde OW=W_{\mathrm{qutrit}}^\dagger\bigl(J_{[n]}^\dagger\widetilde OJ_{[n]}\bigr)W_{\mathrm{qutrit}}=I_{\mc{L}}\otimes A_{\widetilde O}.
	\end{equation}
	The next claim identifies every repaired sector under this same isometry.
	
	\begin{claim}\label{claim:band-sector-factor}
		For every $F\subset[n]$ with $|F|\le e$, there is a subspace $\mc{G}_F\subset\mc{G}_{e\Delta}^{\mathrm{qutrit}}$ such that
		\begin{equation}
			R_F\bigl(\mc{H}_F\cap\im Q_e\bigr)=W(\mc{L}\otimes\mc{G}_F).
		\end{equation}
	\end{claim}
	
	The proof below identifies $\mc{G}_F$ by the conditions that the sites in $F$ are fixed to $\ket{00}$ and that the surviving energy is at most $e-|F|$. Using $R_F^2=I$, define
	\begin{align}
		\mc{G}_e&\coloneqq\bigoplus_{F\subset[n],\,|F|\le e}\mc{G}_F,\\
		W_e\big|_{\mc{L}\otimes\mc{G}_F}&\coloneqq R_FW\big|_{\mc{L}\otimes\mc{G}_F}.
	\end{align}
	The direct sum keeps separate copies of $\mc{G}_F$ even when these subspaces overlap inside $\mc{G}_{e\Delta}^{\mathrm{qutrit}}$. Each summand maps isometrically onto $\mc{H}_F\cap\im Q_e$, and these images are orthogonal.
	\comments{
		For $\ket{z}=\bigoplus_F\ket{z_F}$, this gives
		\begin{equation}
			\|W_e\ket{z}\|^2=\sum_F\|R_FW\ket{z_F}\|^2=\sum_F\|\ket{z_F}\|^2=\|\ket{z}\|^2.
		\end{equation}
		Thus $W_e$ is an isometry with
		\begin{equation}
			\im W_e=\bigoplus_{F\subset[n],\,|F|\le e}\bigl(\mc{H}_F\cap\im Q_e\bigr)=\im Q_e.
		\end{equation}
	}

	Thus $W_e$ is an isometry with image $\im Q_e$.

	Since $\ker H_n\subset\im Q_e$ has dimension $3^k$, the auxiliary space $\mc{G}_e$ is nonzero.
	
	It remains to control operators between sectors indexed by different sets of illegal sites. If $O$ is supported on at most $u$ qubits, then $R_F^\dagger O R_{F'}$ touches at most $|F|+u+|F'|\le2e+u=q$ sites. Equation~\eqref{equation:band-legal-compression} therefore gives $W^\dagger R_F^\dagger O R_{F'}W=I_{\mc{L}}\otimes B_{F,F',O}$. Let $\iota_F:\mc{G}_F\to\mc{G}_{e\Delta}^{\mathrm{qutrit}}$ be the inclusion. The block with input label $F'$ and output label $F$ is
	\begin{align}
		\bigl(W_e^\dagger OW_e\bigr)_{F,F'}&=(I_{\mc{L}}\otimes\iota_F^\dagger)W^\dagger R_F^\dagger O R_{F'}W(I_{\mc{L}}\otimes\iota_{F'}),\\
		&=I_{\mc{L}}\otimes\bigl(\iota_F^\dagger B_{F,F',O}\iota_{F'}\bigr).
	\end{align}
	Defining $A_{e,O}$ by these auxiliary matrix blocks proves Eq.~\eqref{equation:band-full-compression} on the full direct sum.
\end{proof}

\comments{
	\begin{proof}[Proof of Claim~\ref{claim:band-repair}]
		Take $\ket{v_F}\in\mc{H}_F\cap\im Q_e$ with $|F|\le e$. The vector $R_F\ket{v_F}$ is legal and, because $R_F$ commutes with $H_n^{\overline F}$, lies in its spectral subspace with eigenvalues at most $e-|F|$. On the legal space, $H_n=H_n^{\overline F}+H_n^{\cap F}$. These two operators commute and preserve that space, so write the orthogonal joint spectral decomposition
		\begin{equation}
			R_F\ket{v_F}=\sum_{a,d}\ket{w_{a,d}},\quad H_n^{\overline F}\ket{w_{a,d}}=a\ket{w_{a,d}},\quad H_n^{\cap F}\ket{w_{a,d}}=d\ket{w_{a,d}}.
		\end{equation}
		The components need not be normalized. The joint spectral projections preserve the legal space and the spectral subspace of $H_n^{\overline F}$ with eigenvalues at most $e-|F|$, so every nonzero component is legal and has $a\le e-|F|$. This controls spectral support, rather than only the mean energy of the repaired vector. At most $|F|\Delta$ checks touch $F$, giving $0\le H_n^{\cap F}\le|F|\Delta I$ and hence $d\le|F|\Delta$. Each such component therefore has total energy
		\begin{equation}
			a+d\le e-|F|+|F|\Delta=e+(\Delta-1)|F|\le e\Delta.
		\end{equation}
		The last inequality uses $|F|\le e$ and $\Delta\ge1$. Thus the entire repaired vector lies in the legal spectral subspace $\mc{Q}_{e\Delta}^{\mathrm{legal}}$, proving Eq.~\eqref{equation:band-repair-inclusion}.
	\end{proof}
}

\begin{proof}[Proof of Claim~\ref{claim:band-repair}]
	Take $\ket{v_F}\in\mc{H}_F\cap\im Q_e$ with $|F|\le e$. The vector $R_F\ket{v_F}$ is legal and, because $R_F$ commutes with $H_n^{\overline F}$, lies in its spectral subspace with eigenvalues at most $e-|F|$. On the legal space, $H_n=H_n^{\overline F}+H_n^{\cap F}$, the two summands commute, and $0\le H_n^{\cap F}\le|F|\Delta I$ because at most $|F|\Delta$ checks touch $F$. Their joint eigenspaces therefore give
	\begin{equation}
		R_F\bigl(\mc{H}_F\cap\im Q_e\bigr)\subseteq\mc{Q}_{e-|F|+|F|\Delta}^{\mathrm{legal}}\subseteq\mc{Q}_{e\Delta}^{\mathrm{legal}},
	\end{equation}
	where the last inclusion uses $|F|\le e$ and $\Delta\ge1$.
\end{proof}

\begin{proof}[Proof of Claim~\ref{claim:band-sector-factor}]
	Fix $F\subset[n]$ with $|F|\le e$. Within the range of $W$, the repaired sector is specified by being fixed to $\ket{00}$ on $F$ and having surviving energy at most $e-|F|$.
	
	The condition that the sites in $F$ are fixed to $\ket{00}$ uses the projector
	\begin{equation}
		P_F^0\coloneqq\bigotimes_{v\in F}\ketbra{00}{00}_v\otimes I_{F^c}.
	\end{equation}
	Its support has $|F|\le q$ sites, so Eq.~\eqref{equation:band-legal-compression} gives $W^\dagger P_F^0W=I_{\mc{L}}\otimes p_F$ with $0\le p_F\le I$. For any $\ket z\in\mc L\otimes\mc G_{e\Delta}^{\mathrm{qutrit}}$,
	\begin{equation}
		\|(I-P_F^0)W\ket z\|^2=\bra z[I_{\mc L}\otimes(I-p_F)]\ket z.
	\end{equation}
	Thus this condition selects $\mc L\otimes\ker(I-p_F)$, without assuming that $p_F$ is a projector.
	
	The surviving energy is diagonal in the syndrome basis. Define
	\begin{equation}
		D_F\ket{s}\coloneqq|\{i:R_i\cap F=\varnothing,\ s_i\ne0\}|\ket{s},\quad s\in\Xi_{e\Delta}.
	\end{equation}
	Since $W(\mc L\otimes\ket s)=J_{[n]}\mc C_s$, we have $H_n^{\overline F}W=W(I_{\mc L}\otimes D_F)$. Both conditions therefore hold exactly on $\mc L\otimes\mc G_F$, where
	\begin{equation}
		\mc{G}_F\coloneqq\ker(I-p_F)\cap\im\mathbf1_{[0,e-|F|]}(D_F).
	\end{equation}
	For $\ket z\in\mc L\otimes\mc G_F$, the vector $R_FW\ket z$ has illegal set $F$ and, by Eq.~\eqref{equation:ham-sector} and $[R_F,H_n^{\overline F}]=0$, lies in $\im Q_e$. Conversely, Claim~\ref{claim:band-repair} places every repaired vector in the range of $W$, where it satisfies exactly the two conditions defining $\mc G_F$.
\end{proof}

Applying Lemma~\ref{lemma:stabcompletion} to qubits gives the following corollary.

\begin{corollary}[Binary stabilizer completion]\label{corollary:completion-full}
Let $H'=\sum_i O_i$ be a finite sum of Hermitian operators on $M$ qubits, with $O_i$ supported within $A_i$, and let $\ket{S_0}$ be a pure stabilizer state on the same system.
There are $s\le|\bigcup_iA_i|$ independent commuting signed Pauli strings, each supported within some $A_i$, whose common $+1$ projector $P_*$ has rank $2^{M-s}$.
Every operator supported within any $A_i$ compresses to a scalar on $\im P_*$, and
\begin{equation}\label{equation:ref-energy-conclusion}
P_*H'P_*=E_*P_*,\quad E_*\le\bra{S_0}H'\ket{S_0}.
\end{equation}
\end{corollary}

\begin{corollary}\label{corollary:ref-band-weight}
Suppose $H'\geq0$ has integer spectrum and $P\ne0$ is a projector satisfying $PH'P=E_*P$.
For an integer $e\ge0$, let $Q'_e\coloneqq\mathbf1_{[0,e]}(H')$.
Then
\begin{equation}
\|P(I-Q'_e)P\|_\infty\le\frac{E_*}{e+1}.
\end{equation}
Under these spectral assumptions, the projectors in Lemma~\ref{lemma:stabcompletion} and Corollary~\ref{corollary:completion-full} have leakage at most $\bra{S_0}H'\ket{S_0}/(e+1)$.
\end{corollary}

\begin{proof}
The spectrum gives $H'\geq(e+1)(I-Q'_e)$.
Compression by $P$ yields
\begin{equation}
0\leq P(I-Q'_e)P\leq\frac{E_*}{e+1}P.
\end{equation}
Taking the operator norm proves the bound.
The final claim follows from the completion energy inequalities, including Eq.~\eqref{equation:ref-energy-conclusion}.
\end{proof}

We can now prove Theorem~\ref{theorem:main} by combining the logical protection established above for the spectral band with the shared results of Appendix~\ref{appendix:refinement}. This argument uses the energy barrier without assuming operator soundness.

Fix the Hamiltonian $H_n$ in Eq.~\eqref{equation:main-hamiltonian} for an $n$-qutrit stabilizer code with $k\ge1$, term weights at most $\ell\in\mbb{Z}_{>0}$, check incidence at most $\Delta\in\mbb{Z}_{>0}$, and energy barrier $B_n$. Consider any pure binary stabilizer seed $\ket{S}$ on a finite system of $M\ge2n$ qubits and a depth-$D$ circuit $U$ of arbitrary one- and two-qubit gates with disjoint supports in each layer. Let $\rho\coloneqq\Tr_{\mc{A}}(U\ketbra{S}{S}U^\dagger)$ be the retained state, and for an integer $e\ge0$ define
\begin{equation}
	H'\coloneqq U^\dagger(H_n\otimes I_{\mc{A}})U,\quad Q'\coloneqq U^\dagger(\mathbf1_{[0,e]}(H_n)\otimes I_{\mc{A}})U.
\end{equation}
Each pulled-back term is supported within its backward light cone $A_i$, where
\begin{equation}
	|A_i|\le2^D\ell,\quad m\coloneqq\left|\bigcup_iA_i\right|\le\min\{M,2^{D+1}n\}.
\end{equation}
Corollary~\ref{corollary:completion-full} gives $s\le m$ independent commuting signed Pauli generators $\kappa_i$, each supported within one $A_i$, and their common $+1$ projector $P_*$ with rank $2^{M-s}$ and $P_*H'P_*=E_*P_*$, where $E_*\le\Tr(H_n\rho)$. Set
\begin{equation}
	H_K\coloneqq\sum_{i=1}^s\frac{I-\kappa_i}{2},\quad t_m\coloneqq\left\lceil\frac{\sqrt m}{2}\ln8\right\rceil.
\end{equation}

\begin{proposition}[Energy lower bound from logical protection]\label{proposition:sp-energy-bound}
	For the state $\rho$ and projector $P_*$ above, suppose there is an integer $u$ such that
	\begin{equation}\label{equation:sp-joint-hypotheses}
		B_n>[3e+\lfloor u/2\rfloor]\Delta,\quad u\ge2^{2D}\ell t_m.
	\end{equation}
	Then $\Tr(H_n\rho)\ge E_*\ge(e+1)/2$.
\end{proposition}

\begin{proof}
	Fact~\ref{fact:stabilizer-filter} with $a=1/4$ gives a real polynomial $p$ of degree at most $t_m$ satisfying $\|p(H_K)-P_*\|_\infty\le1/4$. Every term in its expansion has input support at most $2^D\ell t_m$ and, after conjugation by $U$, output support at most $2^{2D}\ell t_m\le u$. Expand each output term in tensor products on its retained and discarded supports. Theorem~\ref{theorem:band-full} applies to every retained factor with the same isometry, so
	\begin{equation}
		Q'p(H_K)Q'\big|_{\im Q'}\simeq I_{3^k}\otimes B,\quad B=B^\dagger.
	\end{equation}
	Here the auxiliary space also includes $\mc{A}$. The polynomial is formed before compression, without assuming that $H_K$ or $P_*$ preserves $\im Q'$. Since $3^k$ does not divide $2^{M-s}$, Lemma~\ref{lemma:spectral-rank-obstruction} and Corollary~\ref{corollary:ref-band-weight} give
	\begin{equation}
		\frac12\le\|P_*(I-Q')P_*\|_\infty\le\frac{E_*}{e+1}\le\frac{\Tr(H_n\rho)}{e+1}.
	\end{equation}
\end{proof}

\begin{proof}[Proof of Theorem~\ref{theorem:main}]
	After replacing the barrier constant $\beta_0$ by $\min\{\beta_0,\Delta\}$ if necessary, choose
	\begin{align}
		0<\eta_0&\le\min\left\{\frac14,\frac{\beta_0}{8(3+\ell/2)\Delta}\right\},\quad \eta\coloneqq\frac{\eta_0}{4},\quad e\coloneqq\lfloor\eta_0n\rfloor,\\
		r_*&\coloneqq\frac{\beta_0}{2\Delta},\quad K_{\rm deg}\coloneqq1+\frac{\ln8}{\sqrt2},\quad c\coloneqq\left(\frac{r_*}{\ell K_{\rm deg}}\right)^{2/5}.
	\end{align}
	These constants depend only on the Hamiltonian family. Let $n_B$ be large enough that $B_n\ge\beta_0n$ for family sizes $n\ge n_B$, and restrict to $n\ge n_*\coloneqq\max\{n_B,1,\lceil4\ell\Delta/\beta_0\rceil\}$. The bound on the size of the active input subsystem gives $t_m\le K_{\rm deg}\sqrt n\,2^{D/2}$. If $2^D\le cn^{1/5}$, the integer $u\coloneqq2^{2D}\ell t_m$ satisfies
	\begin{align}
		u&\le\ell K_{\rm deg}\sqrt n\,2^{5D/2}\le r_*n,\\
		[3e+\lfloor u/2\rfloor]\Delta&\le3\eta_0n\Delta+\frac{r_*n\Delta}{2}\le\frac{3\beta_0n}{8}<\beta_0n\le B_n.
	\end{align}
	Thus Eq.~\eqref{equation:sp-joint-hypotheses} holds, and Proposition~\ref{proposition:sp-energy-bound} gives
	\begin{equation}
		\Tr(H_n\rho)\ge\frac{\lfloor\eta_0n\rfloor+1}{2}>\frac{\eta_0n}{2}=2\eta n.
	\end{equation}
	\comments{
		This common bound applies to every branch of depth at most $D$, independently of the seed, circuit, and finite ancillary size. By linearity, the bound holds for every finite mixture of these branches, and hence for every state with $\Cmagicd{2}(\rho)\le D$. The floor gives a uniform strict margin at each fixed size, so the inequality also survives taking an infimum. Taking the contrapositive proves the claimed bound $\Cmagicd{2}(\rho)>\tfrac15\log n+\log c$ whenever $\Tr(H_n\rho)/(2n)\le\eta$. More explicitly, for every integer $D\ge0$ one may take $n_0(D)\coloneqq\max\{n_*,\lceil(2^D/c)^5\rceil\}$, so that every family size $n\ge n_0(D)$ satisfies
		\begin{equation}
			\inf_{\rho:\,\Cmagicd{2}(\rho)\le D}\frac{\Tr(H_n\rho)}{2n}\ge\frac{\lfloor\eta_0n\rfloor+1}{4n}>\eta.
		\end{equation}
		The threshold $\eta$ is therefore independent of $D$.
	}

	This common bound applies to every branch of depth at most $D$, independently of the seed, circuit, and finite ancillary size. By linearity, the bound holds for every finite mixture of these branches, and hence for every state with $\Cmagicd{2}(\rho)\le D$. Taking the contrapositive proves the claimed bound $\Cmagicd{2}(\rho)>\tfrac15\log n+\log c$ whenever $\Tr(H_n\rho)/(2n)\le\eta$.
\end{proof}

\subsection{Separation of the preparation models and other properties}\label{appendix:ternary-properties}

The ternary construction satisfies binary NLTM but does not satisfy intrinsic NLTM. For a ternary stabilizer state $\ket{S_3}$ in the code space, the embedded ground state $\ket{\psi}=J^{\otimes n}\ket{S_3}$ obeys the binary lower bound but has $\Cmagicd{3}(\ketbra{\psi}{\psi})\le2$. Two successive gates between a qutrit and a qubit transfer its high and low binary digits to qubits initialized in $\ket{00}$ and reset the qutrit to $\ket{0}$. Applying these gates in parallel implements the embedding in depth two.

\comments{
	The Tanner realization also has an exact symmetry under charge conjugation. The qutrit charge conjugation $\mathrm C\ket{j}=\ket{-j\bmod3}$ obeys $\mathrm C X\mathrm C^\dagger=X^{-1}$ and $\mathrm C Z\mathrm C^\dagger=Z^{-1}$. For the Tanner CSS realization, every listed check is of pure $X$ or pure $Z$ type, so $\mathrm C^{\otimes n}S_i(\mathrm C^\dagger)^{\otimes n}=S_i^{-1}$ and every $h_i$ in Eq.~\eqref{equation:bar-check-projector} is invariant. Swapping the two qubits implements $\mathrm C$ on the legal subspace and fixes $\ket{11}$. Consequently $\widehat U_C\coloneqq\bigotimes_{v=1}^n\mathrm{SWAP}_v$ commutes with every $\widehat h_i$ and $\ketbra{11}{11}_v$ in this realization, giving an exact global $\mbb{Z}_2$ symmetry with an onsite implementation on the paired sites. The relation of this symmetry to non-Abelian extensions~\cite[Sec.~III.4.3]{christos2026nonabelianquantumlowdensityparity} gives a separate direction for constructing further examples.
}

The operator $\mathrm C\ket{j}=\ket{-j\bmod3}$ obeys $\mathrm C X\mathrm C^\dagger=X^{-1}$ and $\mathrm C Z\mathrm C^\dagger=Z^{-1}$. For the Tanner CSS realization, every listed check is of pure $X$ or pure $Z$ type, so $\mathrm C^{\otimes n}S_i(\mathrm C^\dagger)^{\otimes n}=S_i^{-1}$ and every $h_i$ in Eq.~\eqref{equation:bar-check-projector} is invariant. Swapping the two qubits implements $\mathrm C$ on the legal subspace and fixes $\ket{11}$. Consequently $\widehat U_C\coloneqq\bigotimes_{v=1}^n\mathrm{SWAP}_v$ commutes with every $\widehat h_i$ and $\ketbra{11}{11}_v$ in this realization.

\section{Classical energy witnesses}\label{appendix:classical-witnesses}
\comments{don't use "branch". I prefer to move this section after appendix "Gauged sheaf codes" or "binary NLTM"}
We extend the classical energy witness argument discussed in Refs.~\cite{natarajan2025lecturequantumpcpconjecture,wei2025longrangenonstabilizernessquantumcodes,parham2025quantumcircuitlowerbounds} to unrestricted finite ancillary inputs and classical mixtures.
We first consider the binary specialization of Eq.~\eqref{equation:intrinsic-preparation}, with every input dimension $d_v=2$.


\begin{proposition}\label{proposition:ancilla-witness}
Let $H$ be a Hamiltonian on $N$ qubits, and let $\rho$ be obtained from a pure binary stabilizer state on an arbitrary finite qubit system by a depth-$D$ circuit and discarding all but $N$ outputs.
There exists a state with energy at most $\Tr(H\rho)$ prepared at depth at most $D$ from a pure binary stabilizer state on at most $N2^D$ qubits, followed by discarding outputs. 
\end{proposition}
\begin{proof}
The union $B$ of the backward light cones of the retained outputs has size $L\coloneqq|B|\le N2^D$.
Deleting gates outside these cones leaves a circuit $V$ of depth at most $D$ with $O(LD)$ gates, acting on the reduced input
\begin{equation}
\sigma_B=2^{-L}\sum_{P\in\mc{G}_B}P,
\end{equation}
where $\mc{G}_B$ consists of the input stabilizers supported within $B$, restricted to $B$.
This normalized stabilizer projector has an $O(L^2)$-bit tableau, and completing its generators to a maximal commuting Pauli set expresses it as a mixture of pure stabilizer states.
At least one component has output energy no greater than $\Tr(H\rho)$.
\end{proof}
Proposition~\ref{proposition:ancilla-witness} apply to every finite mixture of such preparations with branch depths at most $D$.
The resulting witness requires neither a description of the mixing probabilities nor an efficient procedure for finding the selected branch.

For verification, write $H=\sum_{i=1}^{m}h_i$, where every $h_i$ has support at most $\ell$ and norm at most one.
The witness specifies the compressed tableau and circuit $V$.
Each $V^\dagger(h_i\otimes I)V$ acts on at most $\ell2^D$ inputs and expands into at most $4^{\ell2^D}$ Pauli operators, whose expectations follow from binary linear algebra~\cite{aaronson2008improvedsimulationstabilizercircuits}.
Writing $L_{\mathrm{bit}}$ for the bit length of the Hamiltonian and circuit descriptions, the energy can be estimated to additive error $0<\epsilon<1$ in time
\begin{equation}
\operatorname{poly}(N,m,L,D,L_{\mathrm{bit}},\ln(1/\epsilon))\,2^{O(\ell2^D)}.
\end{equation}
For fixed $\ell$, polynomial input length, and inverse polynomial $\epsilon$, this is polynomial whenever $2^D=O(\ln N)$.

To ensure finite descriptions of states, encode each gate as a fixed number of two-level rotations and diagonal phases with dyadic angles, whose product is exactly unitary and acts on the same support.
If $g$ is the number of gates, approximating each to operator norm error $\epsilon/(4m\max\{g,1\})$ changes the total energy by at most $\epsilon/2$ and uses $O(\ln(m\max\{g,1\}/\epsilon))$ bits per angle.
The verifier checks independence and commutation of the $L$ signed Hermitian Pauli generators and disjoint gate supports within each layer, then evaluates the energy with controlled numerical error.

For qudit inputs, let $A_i$ be the backward light cone of the support of $h_i$.
It contains at most $\ell2^D$ input qudits, and its generalized Pauli basis has
\begin{equation}
\prod_{v\in A_i}d_v^2\le d_{\max}^{2\ell2^D}
\end{equation}
elements.
Pauli expectations follow by solving the modular equations for membership in the stabilizer subgroup and retaining the associated phases, also for unequal or composite local dimensions.
With explicit gate descriptions, the corresponding energy evaluation takes time
\begin{equation}
\operatorname{poly}(N,m,L,D,L_{\mathrm{bit}},\ln d_{\max},\ln(1/\epsilon))\,d_{\max}^{O(\ell2^D)}.
\end{equation}
Thus fixed depth and a fixed local dimension bound give polynomial time.
More generally, for polynomial description length and inverse polynomial $\epsilon$, the same estimate is polynomial whenever $\ell2^D\ln d_{\max}=O(\ln N)$.
If $d_{\max}$ grows, a general two-qudit gate requires $O(d_{\max}^4)$ two-level rotations and phases, whose description and evaluation costs are included in this bound.

Proposition~\ref{proposition:ancilla-witness} also extends to qudits: the reduced input on $B$ is a normalized qudit stabilizer projector, and extending its commuting Pauli subgroup to a maximal one decomposes it into pure qudit stabilizer states.
One component has no greater output energy, and the same selection applies to classical mixture.
At fixed local dimension bound and depth, the resulting stabilizer generators and circuit have polynomial descriptions, so these preparations likewise admit polynomial classical energy witnesses.

Thus these preparations provide polynomial-size classical witnesses for low energy.
If every YES instance of a $\mathsf{QMA}$-hard local Hamiltonian problem admitted
such a witness, then the problem would lie in \(\mathsf{NP}\), implying
\(\mathsf{QMA}=\mathsf{NP}\).

\printbibliography

\end{document}